\documentclass[lettersize,journal]{IEEEtran}
\usepackage{amsmath,amsfonts}
\usepackage{array}
\usepackage{textcomp}
\usepackage{stfloats}
\usepackage{url}
\usepackage{verbatim}
\usepackage{graphicx}
\usepackage{cite}
\usepackage{epstopdf}
\usepackage{amssymb}
\usepackage{algorithm}
\usepackage{algpseudocode}
\usepackage{multirow}
\usepackage[caption=false,font=footnotesize,labelfont=rm,textfont=rm]{subfig}
\usepackage{xcolor}
\usepackage[hidelinks]{hyperref} 

\newtheorem{theorem}{Theorem}
\newtheorem{proposition}{Proposition}
\newtheorem{remark}{Remark}
\newtheorem{lemma}{Lemma}

\begin{document}

\title{CRB-Based Waveform Optimization for MIMO ISAC Systems With One-Bit ADCs}

\author{
	Qi~Lin,~\IEEEmembership{Graduate Student Member,~IEEE},
	Hong~Shen,~\IEEEmembership{Member,~IEEE},
	Wei~Xu,~\IEEEmembership{Fellow,~IEEE}, 
	and Chunming~Zhao,~\IEEEmembership{Member,~IEEE}
	
	\thanks{This work was supported by the National Natural Science Foundation of China under Grant 62271137. Part of this work has been presented at
		the 2025 IEEE 101st Vehicular Technology Conference (VTC2025-Spring), Oslo, Norway, Jun. 2025 \cite{ref0}. \emph{(Corresponding author: Hong Shen.)}}
	\thanks{Qi~Lin, Hong~Shen, and Chunming~Zhao are with the National Mobile Communications Research Laboratory, Southeast University, Nanjing 210096, China (e-mail: qilin@seu.edu.cn; shhseu@seu.edu.cn; cmzhao@seu.edu.cn).}
	\thanks{Wei Xu is with the National Mobile Communications Research Laboratory, Southeast University, Nanjing 210096, China, and also with the Purple Mountain Laboratories, Nanjing 211111, China (e-mail: wxu@seu.edu.cn).}                      
}

\maketitle

\begin{abstract}
This paper studies the transmit waveform optimization for a quantized multiple-input multiple-output (MIMO) integrated sensing and communication (ISAC) system, where one-bit analog-to-digital converters (ADCs) are employed to enable a low-cost and power-efficient hardware implementation. Focusing on the parameter estimation task, we propose two novel Cram\'er-Rao bounds (CRBs) for both point-like target (PT) and extended target (ET) to characterize the impact of quantization distortion on the estimation accuracy, where associated estimation methods are also developed to approach these theoretical CRBs. Moreover, with the goal of jointly enhancing the sensing and communication performances, we formulate the bi-criterion ISAC waveform optimization problem by minimizing the derived CRB objectives subject to a communication symbol error probability (SEP) constraint and a total power constraint, which, due to the high nonlinearity of the one-bit CRBs, are extremely nonconvex. To yield a high-quality suboptimal solution, we develop an effective alternating direction method of multipliers (ADMM) framework which exploits the majorization-minimization (MM) technique to address the nonconvex issue. Simulation results verify that the one-bit CRBs are tight for characterizing the quantized estimation performance and the proposed estimation methods also show clear performance advantages over the existing benchmark schemes. Furthermore, a flexible trade-off between the CRB and the SEP performance can be achieved by the developed ADMM framework, demonstrating the effectiveness of the optimized ISAC waveform.
\end{abstract}

\begin{IEEEkeywords}
Integrated sensing and communication (ISAC), one-bit analog-to-digital converters (ADCs), Cram\'er-Rao bound (CRB), waveform optimization.
\end{IEEEkeywords}

\section{Introduction}
\IEEEPARstart{I}{ntegrated} sensing and communication (ISAC) empowers current wireless communication systems to support various  sensing scenarios and has become an important research focus \cite{ref1,ref2,ref3,ref4}. Particularly, in the recent IMT-2030 framework \cite{ref5}, ISAC has been envisioned as a pivotal enabler for the sixth-generation (6G) wireless systems and expected to play a key role in providing full sensing capabilities while meeting crucial communication requirements.

To fully reap the benefits of ISAC, multiple-input multiple-output (MIMO) based transceiver optimization has been widely investigated in a number of recent studies \cite{ref21,ref22,ref23,ref24,ref25,ref26,ref27,ref28}. Specifically, given the fact that the sensing capability heavily hinges on the beampattern of transmitted signals, the authors of \cite{ref21,ref22,ref23} optimized the transmit beampattern with the communication signal-to-interference-plus-noise ratio (SINR) \cite{ref21,ref22} and the symbol error rate (SER) \cite{ref23} constraints imposed, respectively. Moreover, focusing on improving the parameter estimation accuracy, the authors of \cite{ref24,ref25,ref26} advocated the minimization of the Cram\'er-Rao bound (CRB) \cite{ref24,ref25} or the posterior CRB \cite{ref26} objectives, which represents the lower bound of the mean-squared error (MSE) of arbitrary unbiased estimators for deterministic or stochastic unknown parameters. More specifically, the authors of \cite{ref24,ref26} optimized the transmit covariance by minimizing the CRB and the posterior CRB, respectively, subject to the communication SINR and total power constraints, while the authors of \cite{ref25} optimized the transmit covariance by maximizing the communication achievable rate subject to the maximum allowable CRB and total power constraints. Distinct from the above studies, the target detection probability was investigated in \cite{ref27,ref28} and the associated sensing SINR was maximized subject to the communication achievable rate and various power constraints.

Effective implementation of MIMO ISAC systems requires a large antenna array to provide reliable sensing and communication capabilities, which inevitably increases the hardware cost and power consumption of the radio frequency (RF) chains at the ISAC transceiver. A practical solution to address this issue is using low-resolution (i.e., few quantization bits) digital-to-analog converters (DACs) and analog-to-digital converters (ADCs). Notably, deploying one-bit quantizers yields unparalleled cost and power efficiency, which is thus the focus of this paper. Numerous transceiver designs concerning one-bit quantization have been proposed for MIMO systems, including symbol level precoding (SLP) with one-bit DACs \cite{ref41,ref42,ref43} as well as channel estimation \cite{ref44} and symbol detection \cite{ref45} with one-bit ADCs. In contrast, only a handful of recent studies considered the transceiver optimization for one-bit MIMO ISAC systems \cite{ref46,ref47,ref48}. Specifically, the authors of \cite{ref46,ref47} optimized the transmit waveform for MIMO ISAC systems with one-bit DACs, where the communication MSE is minimized subject to the sensing CRB and binary DAC constraints in \cite{ref46} and a weighted sum of the communication MSE and the sensing waveform similarity was minimized with binary DAC constraints imposed in \cite{ref47}. Furthermore, a joint transceiver design for MIMO ISAC systems with one-bit DACs and ADCs was proposed in \cite{ref48}, where the communication MSE and a quantized SINR metric were jointly optimized to enhance both the downlink communication and target detection performances. However, while the SINR metric can effectively characterize the target detection performance, it fails to evaluate the accuracy of parameter estimation, which constitutes another crucial sensing task. Consequently, optimizing the parameter estimation performance inherently calls for a new objective function, such as the CRB, and the performance trade-off between parameter estimation and communication under one-bit quantization should also be carefully revealed.
Nonetheless, the problem of improving the parameter estimation performance for MIMO ISAC systems with one-bit ADCs remains to be explored. A closely related work is \cite{ref49}, wherein the authors presented an insightful CRB metric based on binary observations to characterize the estimation accuracy for one-bit MIMO radar. However, the proposed one-bit CRB in \cite{ref49} involves a nonanalytical Q-function, which discourages its use as a transceiver design metric. In summary, distinct from the aforementioned studies [20]–[23], a mathematically tractable CRB metric explicitly tailored for parameter estimation, which can be utilized to address the highly coupled transceiver design under one-bit ADCs, is highly desired. To the best of the authors' knowledge, the waveform design problem aimed at improving both the parameter estimation and the downlink multiuser communication for MIMO ISAC systems with one-bit ADCs has not been investigated before, which motivates this work.

In this paper, we study the parameter estimation and waveform optimization for MIMO ISAC systems equipped with one-bit ADCs to achieve superior energy and hardware efficiency. However, the severe one-bit quantization distortions yield binary observations at the ISAC receiver, which complicate the theoretical performance analysis and the corresponding ISAC waveform design. The main contributions of this paper are summarized as follows:
\begin{itemize}
\item Through a Bussgang-based quantization analysis and by leveraging the worst-case Gaussian assumption, we derive two novel CRB metrics for both point-like target (PT) and extended target (ET) models to characterize their estimation accuracy under one-bit quantization. Moreover, we also present the corresponding one-bit estimation methods to approach the derived CRBs.
\item By minimizing the proposed CRB objectives subject to a communication symbol error probability (SEP) constraint and a total power constraint, we formulate new ISAC waveform design problems for the PT and ET scenarios, respectively, each of which turns out to be extremely nonconvex. To find a high-quality solution, we develop an effective alternating direction method of multipliers (ADMM) based algorithm, in which the majorization-minimization (MM) technique is exploited to solve the subproblems in each iteration.
\item Numerical results are presented to validate the tightness of the derived one-bit CRBs, and to show that the proposed waveform design yields noticeable performance improvements compared to the existing benchmark schemes, despite at the cost of an increased computational complexity due to the quantization-aware optimization. Moreover, the developed ADMM framework also achieves a flexible trade-off between the estimation and communication performances for both PT and ET scenarios, thereby verifying the effectiveness of the optimized ISAC waveform.
\end{itemize}

The rest of the paper is organized as follows.
Section \ref{section II} introduces the considered MIMO ISAC system model with one-bit ADCs.
In Section \ref{section III}, we derive one-bit CRB metrics for both PT and ET, along with their respective one-bit estimation methods.
Sections \ref{section IV} and \ref{section V} elaborate on the ISAC waveform optimization algorithms for PT and ET scenarios, respectively.
Section \ref{section VI} presents the simulation results and conclusions are drawn in Section \ref{section VII}.

\emph{Notations:}
Lowercase letters, boldface lowercase letters, and boldface uppercase letters are used to denote scalars, column vectors, and matrices, respectively.
The conjugate, the transpose, the Hermitian transpose, and the inverse of $\mathbf{A}$ are denoted by $\mathbf{A}^*$, $\mathbf{A}^T$, $\mathbf{A}^H$, and $\mathbf{A}^{-1}$, respectively.
The $(i,j)$-th entry of $\mathbf{A}$ is denoted by $[\mathbf{A}]_{i,j}$.
We use $\operatorname{tr}(\mathbf{A})$ and $\operatorname{vec}(\mathbf{A})$ to denote the trace and the vectorization of $\mathbf{A}$, respectively, and $\operatorname{unvec}(\cdot)$ denotes the inverse operation of $\operatorname{vec}(\cdot)$.
We denote by $\lambda_{\max}(\mathbf{A})$ the maximum eigenvalue of $\mathbf{A}$.
The operator $\operatorname{diag}(\cdot)$ signifies a diagonal matrix whose diagonal elements lie on the main diagonals of the input matrix or are the elements of the input vector.
The operator $\frac{\partial(\cdot)}{\partial \zeta}$ denotes the derivative with respect to $\zeta$, which corresponds to the Wirtinger derivative when $\zeta$ is a complex variable.
We denote by $\operatorname{sign}(\cdot)$ the elementwise sign function, and $\Re(\cdot)$ and $\Im(\cdot)$ return the real and the imaginary parts of the input, respectively.
The operator $\operatorname{card}(\mathcal{C})$ represents the cardinality of set $\mathcal{C}$. 
The operator $\Vert\cdot\Vert_2$ denotes the $\ell_2$ norm of vectors.
The operator $\mathbb{E}\left\{\cdot\right\}$ denotes the expectation.
The Kronecker product and the Hadamard product are denoted by $\otimes$ and $\circ$, respectively.
$Q(x)=\frac{1}{\sqrt{2\pi}}\int_x^\infty e^{-\frac{u^2}{2}}\,\operatorname{d}\!u$ denotes the tail distribution function of the standard Gaussian distribution.
We use $\mathbf{a}\sim\mathcal{CN}\left(\mathbf{0},\mathbf{C}_{\mathbf{aa}}\right)$ to specify that $\mathbf{a}$ follows a circularly symmetric complex Gaussian distribution
with zero mean and covariance matrix $\mathbf{C}_{\mathbf{aa}}$.
We denote the identity matrix of size $N$ by  $\mathbf{I}_N$, and the all-ones and all-zeros vectors of size $N$ are denoted by $\mathbf{1}_N$ and $\mathbf{0}_N$, respectively.
We use $\mathbb{R}^{M\times N}$ and $\mathbb{C}^{M\times N}$ to denote the set of real and complex matrices of dimension $M\times N$, respectively.
We denote by $\mathbf{T}_{M,N}\in\mathbb{R}^{MN\times MN}$ the commutation matrix that, for $\mathbf{A}\in\mathbb{C}^{M\times N}$, transforms $\operatorname{vec}(\mathbf{A})$ into $\operatorname{vec}(\mathbf{A}^T)$.
The imaginary unit is denoted by $j$.

\section{System Model}\label{section II}
We consider a MIMO ISAC system, where the transmitter with $N_t$ antennas and the receiver with $N_r$ antennas are collocated to transmit ISAC signals and simultaneously perform target sensing via the received echoes. Moreover, to facilitate a low-cost and power-efficient hardware implementation, we assume that one-bit ADCs are employed at the ISAC receiver. An illustration of the considered quantized MIMO ISAC system is depicted in Fig.~\ref{fig.system_model} at the top of the next page.

\subsection{Communication Model and Performance Metric}\label{subsection II A}
The ISAC transmitter communicates with $K$ single-antenna user equipments (UEs) via a downlink MIMO transmission, as shown in Fig.~\ref{fig.system_model}. Denote the UE information signal by $\mathbf{S}\in\mathcal{S}^{K\times L}$, where $\mathcal{S}$ denotes the set of constellation points and $L$ is the block length. In particular, we consider the commonly used $M$-ary quadrature amplitude modulation (QAM) constellation and accordingly define $\mathcal{S}$ as
\begin{equation}
\mathcal{S}\triangleq \bigg\{\left.s_R+js_I\right\vert s_R,s_I\in\left\{\pm1,\dots,\pm \left(\sqrt{M}-1\right)\right\}\bigg\}.
\end{equation}
Let $\mathbf{X}=\left[\mathbf{x}_1,\dots,\mathbf{x}_L\right]\in\mathbb{C}^{N_t\times L}$ denote the transmit waveform matrix. The corresponding UE received signal can be given by
\begin{equation}\label{UE received signal Y}
\mathbf{Y}=\mathbf{H}\mathbf{X}+\mathbf{W},
\end{equation}
where $\mathbf{H}=\left[\mathbf{h}_1,\dots \mathbf{h}_K\right]^H\in\mathbb{C}^{K\times N_t}$ represents the downlink MIMO channel matrix, with $\mathbf{h}_k^H$ being the channel between the ISAC transmitter and the $k$-th UE, and $\mathbf{W}$ denotes the additive white Gaussian noise (AWGN) at the UE receiver satisfying $\operatorname{vec}\left(\mathbf{W}\right)\sim\mathcal{CN}\left(\mathbf{0}_{KL},\sigma_w^2\mathbf{I}_{KL}\right)$, with $\sigma_w^2$ denoting the noise power. Particularly, we consider a standard Rayleigh block-fading channel model, where $\mathbf{H}$ is assumed to remain constant within each coherence block but varies independently across different intervals.

\begin{figure}[t]
	\centering
	\includegraphics[scale=0.5]{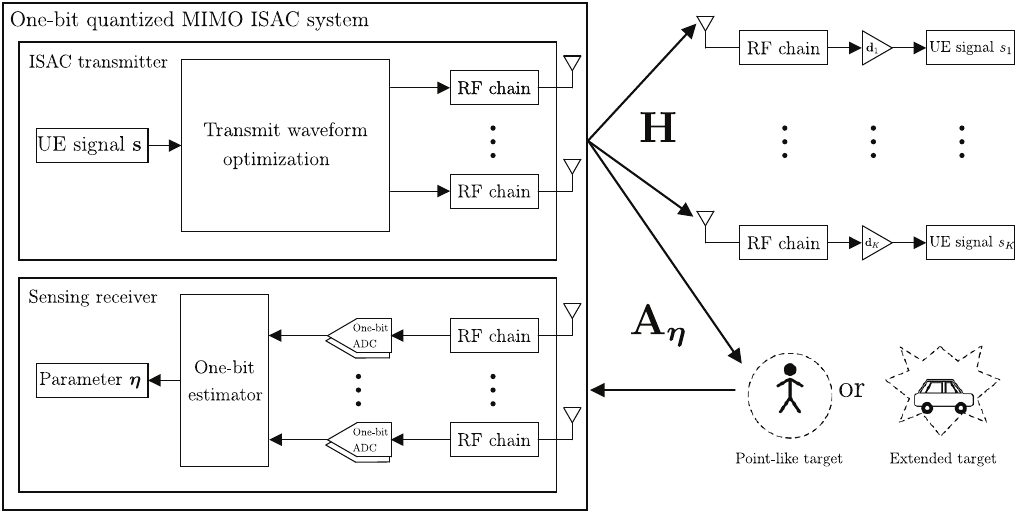}
	\caption{The considered MIMO ISAC system with one-bit ADCs.}
	\label{fig.system_model}
\end{figure}

To improve the detection performance for QAM scheme, each UE utilizes two real decision variables to equalize the real and the imaginary parts of the received signal in \eqref{UE received signal Y}. Specifically, by defining $\mathbf{d}_R=\left[d_{R,1},\dots,d_{R,K}\right]^T\in\mathbb{R}^K$ and $\mathbf{d}_I=\left[d_{I,1},\dots,d_{I,K}\right]^T\in\mathbb{R}^K$ with $(d_{R,k}, d_{I,k})$ being the pair of decision variables of the $k$-th UE, we can express the real and the imaginary parts of the reconstructed information signal $\hat{\mathbf{S}}$ as
\begin{equation}
\begin{bmatrix}
\Re(\hat{\mathbf{S}})\\
\Im(\hat{\mathbf{S}})
\end{bmatrix}
=\operatorname{diag}(\mathbf{d})^{-1}\begin{bmatrix}\Re\left(\mathbf{H}\right)\Re({\mathbf{X}})-\Im\left(\mathbf{H}\right)\Im({\mathbf{X}})+\Re({\mathbf{W}})\\ \Im(\mathbf{H})\Re({\mathbf{X}})+\Re(\mathbf{H})\Im({\mathbf{X}})+\Im({\mathbf{W}})\end{bmatrix} \label{reconstructed signal Shat}
\end{equation}
where we define $\mathbf{d}=\left[\mathbf{d}_R^T,\mathbf{d}_I^T\right]^T\in\mathbb{R}^{2K}$ and exploit the real-valued model of \eqref{UE received signal Y}. From \eqref{reconstructed signal Shat}, we then obtain
\begin{equation}
\begin{aligned}
\Re(\hat{s}_{k,l})&= \Re(\mathbf{h}_k^H\mathbf{x}_l)/d_{R,k}+\Re(w_{k,l})/d_{R,k},\quad k\in\mathcal{K},\ l\in\mathcal{L},\\
\Im(\hat{s}_{k,l})&= \Im(\mathbf{h}_k^H\mathbf{x}_l)/d_{I,k}+\Im(w_{k,l})/d_{I,k},\quad k\in\mathcal{K},\ l\in\mathcal{L},
\end{aligned}
\end{equation}
where $\hat{s}_{k,l}=[\hat{\mathbf{S}}]_{k,l}$ and $\mathcal{K} \triangleq \{1,\dots,K\}$ and $\mathcal{L} \triangleq \{1,\dots,L\}$. To ensure the communication quality-of-service (QoS) of the $k$-th UE, we assume that $\text{SEP}_{k,l}\le \varepsilon_k, l\in\mathcal{L}$, where $\text{SEP}_{k,l}$ denotes the SEP\footnote{In this paper, SEP serves as the theoretical communication metric used for waveform optimization, whereas SER represents the empirical communication performance. The relationship between the SEP metric and the actual SER will be verified via simulations in Section \ref{section VI}.} corresponding to the hard decision of $s_{k,l}$ and $\varepsilon_k$ is the maximum allowable SEP at the $k$-th UE. Clearly, by noting that $\text{SEP}_{k,l}=1-(1-\text{SEP}_{R,k,l})(1-\text{SEP}_{I,k,l})$ with $\text{SEP}_{R,k,l}$ and $\text{SEP}_{I,k,l}$ denoting the SEP associated with the hard decision of $\Re(s_{k,l})$ and $\Im(s_{k,l})$, respectively, the aforementioned constraints can be equally transformed into
\begin{equation}\label{QoS for UE k}
\begin{aligned}
\text{SEP}_{R,k,l},\ \text{SEP}_{I,k,l} \le 1-\sqrt{1-\varepsilon_k},\quad k\in\mathcal{K},\ l\in\mathcal{L}.
\end{aligned}
\end{equation}
Furthermore, it is shown in \cite{ref43,ref102} that the constraints in \eqref{QoS for UE k} amount to the following linear inequality constraints:
\begin{equation}\label{SEP metric in vectorized form}
\begin{aligned}
	-\mathbf{d}_R+\mathbf{b}_{R,l}\le&\Re\left(\mathbf{H}\mathbf{x}_l\right)-\mathbf{d}_R\circ \Re(\mathbf{s}_l) \le \mathbf{d}_R- \mathbf{a}_{R,l},\ l\in\mathcal{L},\\
	-\mathbf{d}_I+\mathbf{b}_{I,l}\le&\Im\left(\mathbf{H}\mathbf{x}_l\right)-\mathbf{d}_I\circ \Im(\mathbf{s}_l) \le \mathbf{d}_I- \mathbf{a}_{I,l},\ l\in\mathcal{L},\\
	\mathbf{d}_R\ge& \boldsymbol{\gamma},\quad \mathbf{d}_I\ge  \boldsymbol{\gamma},
\end{aligned}
\end{equation}
where $\mathbf{a}_{R,l}\triangleq\left[a_{R,1,l},\dots,a_{R,K,l}\right]^T\in\mathbb{R}^K$ and $\mathbf{b}_{R,l}\triangleq\left[b_{R,1,l},\dots,b_{R,K,l}\right]^T\in\mathbb{R}^K$ are constant parameters, with $a_{R,k,l}$ and $b_{R,k,l}$ given by \cite{ref43,ref102}
\begin{align}
a_{R,k,l}&=\left\{
\begin{aligned}
&\frac{\sigma_w}{\sqrt{2}}Q^{-1}\left(\frac{1-\sqrt{1-\varepsilon_k}}{2}\right),&& \hspace{-0.2cm} \vert \Re(s_{k,l})\vert < \sqrt{M}-1,\\
& -\infty,&& \hspace{-0.5cm} \Re(s_{k,l})=\sqrt{M}-1,\\
& \frac{\sigma_w}{\sqrt{2}}Q^{-1}\left(1-\sqrt{1-\varepsilon_k}\right),&& \hspace{-0.5cm} \Re(s_{k,l})=-(\sqrt{M}-1), \label{a_R_k_l}\end{aligned}\right.\\
b_{R,k,l}&=\left\{
\begin{aligned}
&\frac{\sigma_w}{\sqrt{2}}Q^{-1}\left(\frac{1-\sqrt{1-\varepsilon_k}}{2}\right),&& \hspace{-0.2cm} \vert \Re(s_{k,l})\vert < \sqrt{M}-1,\\
& \frac{\sigma_w}{\sqrt{2}}Q^{-1}\left(1-\sqrt{1-\varepsilon_k}\right),&& \hspace{-0.5cm} \Re(s_{k,l})=\sqrt{M}-1,\\
& -\infty,&& \hspace{-0.5cm} \Re(s_{k,l})=-(\sqrt{M}-1),\label{b_R_k_l}\end{aligned}
\right.
\end{align}
$\mathbf{a}_{I,l}$ and $\mathbf{b}_{I,l}$ can be defined in the same way as $\mathbf{a}_{R,l}$ and $\mathbf{b}_{R,l}$, with $a_{I,k,l}$ and $b_{I,k,L}$ obtained by replacing $\Re(\cdot)$ with $\Im(\cdot)$ in \eqref{a_R_k_l} and \eqref{b_R_k_l}, respectively, and $\boldsymbol{\gamma}\triangleq\left[\gamma_1,\dots,\gamma_K\right]^T\in\mathbb{R}^K$ with $\gamma_k=\frac{\sigma_w}{\sqrt{2}}Q^{-1}\left(\frac{1-\sqrt{1-\varepsilon_k}}{2}\right)$ being the minimum decision value associated with the $k$-th UE for inequalities in \eqref{SEP metric in vectorized form} to hold \cite{ref102}. It is worth noting that \eqref{SEP metric in vectorized form} serves as the communication QoS constraints in our proposed ISAC waveform design, as will be shown in Sections \ref{section IV} and \ref{section V}.

\begin{remark}
From a communication perspective, the considered waveform design belongs to the category of nonlinear symbol level precoding (SLP) \cite{ref41,ref42,ref43}, which directly optimizes the baseband signal matrix $\mathbf{X}$ based on specific information symbol blocks using standard modulation schemes (e.g., QAM). Moreover, the waveform design is compatible with practical physical-layer architectures like the cyclic prefix orthogonal frequency division multiplexing (CP-OFDM) \cite{ref105}.
\end{remark}

\subsection{Sensing Model and Bussgang Decomposition}\label{sec II B}
Recall that the ISAC waveform $\mathbf{X}$ is also intended for target sensing, as illustrated in Fig.~\ref{fig.system_model}, and thus the corresponding echo signal at the ISAC receiver is given by
\begin{equation}
	\begin{aligned}
		\mathbf{R}=\mathbf{A}_{\boldsymbol{\eta}}\mathbf{X}+\mathbf{V},\label{echo eignal R}
	\end{aligned}
\end{equation}
where $\mathbf{A}_{\boldsymbol{\eta}}\in\mathbb{C}^{N_r\times N_t}$ represents the target response matrix, dependent on the array geometry and the target parameter $\boldsymbol{\eta}$ to be estimated, and $\mathbf{V}\in\mathbb{C}^{N_r\times L}$ is the receiver AWGN satisfying $\operatorname{vec}\left(\mathbf{V}\right)\sim\mathcal{CN}\left(\mathbf{0}_{N_rL},\sigma_v^2\mathbf{I}_{N_rL}\right)$, with $\sigma_v^2$ denoting the noise power. Note that the above model embraces a variety of realistic scenarios and here we focus on two common target models, i.e., PT and ET \cite{ref24,ref25} and discuss in detail their CRB metrics for measuring the estimation accuracy in the next section. To proceed, we rewrite \eqref{echo eignal R} in a vectorized form as
\begin{equation}
	\mathbf{r}=\left(\mathbf{X}^T\otimes \mathbf{I}_{N_r}\right)\mathbf{a}_{\boldsymbol{\eta}}+\mathbf{v},\label{echo signal r}
\end{equation}
where $\mathbf{r}=\operatorname{vec}(\mathbf{R})$, $\mathbf{a}_{\boldsymbol{\eta}}=\operatorname{vec}\left(\mathbf{A}_{\boldsymbol{\eta}}\right)$, and $\mathbf{v}=\operatorname{vec}(\mathbf{V})$. Due to the use of one-bit ADCs, the quantized received signal associated with \eqref{echo signal r} can be expressed as
\begin{equation}
	\mathbf{z}=\frac{1}{\sqrt{2}}\operatorname{sign}\left(\Re(\mathbf{r})\right)+\frac{j}{\sqrt{2}}\operatorname{sign}\left(\Im(\mathbf{r})\right).\label{quantized signal z}
\end{equation}
To handle the severe nonlinear distortions of one-bit quantization, we employ the Bussgang decomposition \cite{ref103} to obtain a statistically equivalent linear representation of $\mathbf{z}$. Specifically, assuming that the unquantized signal $\mathbf{r}$ is Gaussian distributed with zero mean and covariance matrix $\mathbf{C}_{\mathbf{rr}}$, which will be justified in Section \ref{section III}, the linearized approximation of $\mathbf{z}$ via the Bussgang decomposition, can be given by
\begin{equation}\label{linearized signal z}
	\mathbf{z}\approx \mathbf{F}\mathbf{r}+\mathbf{q},
\end{equation}
where $\mathbf{F}=\sqrt{\frac{2}{\pi}}\operatorname{diag}\left(\mathbf{C}_{\mathbf{rr}}\right)^{-\frac{1}{2}}$
is the real Bussgang gain matrix, and $\mathbf{q}$ is the quantization noise uncorrelated with $\mathbf{r}$ \cite{ref104}. Moreover, as detailed in \cite{ref104}, $\mathbf{z}$ also has zero mean and its covariance matrix is given by
\begin{equation}
	\mathbf{C}_{\mathbf{zz}}=\frac{2}{\pi}\arcsin\left(\operatorname{diag}\left(\mathbf{C}_{\mathbf{rr}}\right)^{-\frac{1}{2}}\mathbf{C}_{\mathbf{rr}}\operatorname{diag}\left(\mathbf{C}_{\mathbf{rr}}\right)^{-\frac{1}{2}}\right).\label{covariance matrix Czz}
\end{equation}

\begin{remark}
It should be mentioned that instantaneous channel state information (CSI) and prior target information (e.g., nominal DOA or target response statistics) are assumed to be known in our formulations. In practical implementations, the requisite prior target parameters can be readily acquired through an initial coarse estimation phase using standard, unoptimized waveforms (e.g., orthogonal DFT beams) \cite{ref24,ref25,ref26}. Meanwhile, high-precision communication CSI can be obtained via dedicated channel estimation techniques.
\end{remark}

\section{One-Bit Parameter Estimation: \\ CRB Analysis and Method Design}\label{section III}
In this section, we first derive novel CRB metrics for both PT and ET to characterize their parameter estimation performance under one-bit quantization by exploiting the Bussgang-based analysis in Section \ref{sec II B}. Then we present practical estimation methods based on the binary observations in \eqref{quantized signal z} for both target models to approach their respective theoretical CRBs.

\subsection{Derivation of One-Bit CRB for Point-Like Targets}\label{subsection III A}
For PT, the target response matrix $\mathbf{A}_{\boldsymbol{\eta}}$ can be given by $\mathbf{A}_{\boldsymbol{\eta}}=\alpha\mathbf{a}_r(\theta)\mathbf{a}_t(\theta)^T$ with $\boldsymbol{\eta}\triangleq[\alpha,\theta]^T$, where $\alpha\in\mathbb{C}$ denotes the reflection coefficient accounting for both the round-trip path loss and the radar cross section (RCS) of the target, and $\mathbf{a}_t(\theta)$ and $\mathbf{a}_r(\theta)$ are the array responses corresponding to the ISAC transmit and receive antennas, respectively, with $\theta$ denoting the direction of arrival (DOA) as well as the direction of departure (DOD) for the considered monostatic configuration. For the commonly used uniform linear array (ULA) with half-wavelength antenna spacing, the expressions for  $\mathbf{a}_t(\theta)$ and $\mathbf{a}_r(\theta)$ are, respectively, given by
\begin{equation}
\begin{aligned}
\mathbf{a}_t(\theta)&=\frac{1}{\sqrt{N_t}}\left[1,e^{-j\pi\sin(\theta)},\dots,e^{-j\pi(N_t-1)\sin(\theta)}\right]^T,\\
\mathbf{a}_r(\theta)&=\frac{1}{\sqrt{N_r}}\left[1,e^{-j\pi\sin(\theta)},\dots,e^{-j\pi(N_r-1)\sin(\theta)}\right]^T.
\end{aligned}
\end{equation}
Using the expressions above, we recast the unquantized echo signal $\mathbf{r}$ in \eqref{echo signal r} as
\begin{equation}
	\mathbf{r}=\alpha\mathbf{A}_\theta \mathbf{x}+\mathbf{v},\label{point-like targets r}
\end{equation}
where $\mathbf{x}=\operatorname{vec}(\mathbf{X})$ and we define 
\begin{equation}
	\mathbf{A}_\theta=\mathbf{I}_L\otimes \left(\mathbf{a}_r(\theta)\mathbf{a}_t(\theta)^T\right),\label{simplified notation Atheta}
\end{equation}
for notational convenience. To account for the complex fluctuations of realistic targets, we model the reflection coefficient as $\alpha\sim\mathcal{CN}(\mu_\alpha,\sigma_\alpha^2)$ \cite{ref200}, where the mean $\mu_\alpha$ represents the reflection determined by the target range and the variance $\sigma_\alpha^2$ characterizes the random variations. Consistent with \cite{ref24,ref25, ref26}, we are interested in estimating the DOA $\theta$, while treating $\alpha$ as a nuisance parameter. To facilitate a tractable CRB expression for $\theta$, we first introduce the following lemma.
\begin{lemma}\label{lemma alpha}
For a parameter of interest $\theta$ and a nuisance parameter $\alpha$, the CRB of $\theta$ derived under a zero-mean assumption for $\alpha$ provides a strict upper bound on the CRB derived under a non-zero mean $\alpha$.
\end{lemma}
\begin{IEEEproof}
	Please refer to \cite{ref201}.
\end{IEEEproof}

Accordingly, we assume a zero-mean $\alpha\sim\mathcal{CN}(0,\sigma_\alpha^2)$, which yields a conservative CRB for $\theta$, and therefore minimizing this conservative bound (as will be seen in Section \ref{section IV}) inherently aligns with a minimax optimization framework. Moreover, it follows from \eqref{point-like targets r} that the unquantized echo signal $\mathbf{r}$ is Gaussian distributed with zero mean, and its covariance matrix $\mathbf{C}_{\mathbf{rr}}$ can be expressed as
\begin{equation}
	\mathbf{C}_{\mathbf{rr}}=\sigma_{\alpha}^2 \mathbf{A}_\theta \mathbf{xx}^H\mathbf{A}_\theta^H+\sigma_v^2 \mathbf{I}_{N_rL}.\label{covariance matrix Crr}
\end{equation}
This justifies the Gaussian assumption required for applying the Bussgang decomposition in \eqref{linearized signal z}. To proceed with the CRB derivation for $\theta$, we further introduce the subsequent lemma.
\begin{lemma}\label{lemma z}
	Although the Bussgang-based observation $\mathbf{z}$ in \eqref{linearized signal z} is not Gaussian, treating it as Gaussian distributed with the exact first- and second-order moments (i.e., zero mean and covariance matrix $\mathbf{C}_{\mathbf{zz}}$ in \eqref{covariance matrix Czz}) yields a worst-case CRB of $\theta$, thereby conforming to the minimax paradigm.
\end{lemma}
\begin{IEEEproof}
	Please refer to \cite{ref49,ref202}.
\end{IEEEproof}
Thus, by invoking Lemmas \ref{lemma alpha} and \ref{lemma z}, we then obtain the worst-case likelihood function for estimating $\theta$, whose CRB expression, according to \cite[Chapter 15]{ref201}, can be given by 
\begin{equation}
	\text{CRB}_\theta=\operatorname{tr}\left(\mathbf{C}_{\mathbf{zz}}^{-1}\frac{\partial \mathbf{C}_{\mathbf{zz}}}{\partial \theta}\mathbf{C}_{\mathbf{zz}}^{-1}\frac{\partial \mathbf{C}_{\mathbf{zz}}}{\partial \theta}\right)^{-1}.\label{CRB with arcsin Czz}
\end{equation}
Note that the above one-bit CRB of $\theta$ involves calculating the inverse and the derivative of $\mathbf{C}_{\mathbf{zz}}$ in \eqref{covariance matrix Czz}, which does not admit a tractable form due to the $\arcsin$ function. Here, we utilize the approximation $\arcsin(x)\approx x, \ x\ne1$ \cite{ref44}, and thus the covariance matrix $\mathbf{C}_{\mathbf{zz}}$ in \eqref{covariance matrix Czz} can be approximated as
\begin{equation}\label{approximated covariance Czz}
	\begin{aligned}	
		\mathbf{C}_{\mathbf{zz}}&\approx \mathbf{F}\mathbf{C}_{\mathbf{rr}}\mathbf{F}+\left(1-\frac{2}{\pi}\right)\mathbf{I}_{N_rL}
		\triangleq\hat{\mathbf{C}}_{\mathbf{zz}}.
	\end{aligned}
\end{equation}
Using this result, we can therefore obtain a more tractable expression for $\text{CRB}_\theta$ given by
\begin{equation}\label{point-like targets CRB}
	\text{CRB}_\theta\approx \operatorname{tr}\left(\hat{\mathbf{C}}_{\mathbf{zz}}^{-1}\frac{\partial \hat{\mathbf{C}}_{\mathbf{zz}}}{\partial \theta}\hat{\mathbf{C}}_{\mathbf{zz}}^{-1}\frac{\partial \hat{\mathbf{C}}_{\mathbf{zz}}}{\partial \theta}\right)^{-1},
\end{equation}
where $\frac{\partial \hat{\mathbf{C}}_{\mathbf{zz}}}{\partial \theta}=\frac{\partial \mathbf{F}}{\partial\theta}\mathbf{C}_{\mathbf{rr}}\mathbf{F}+\mathbf{F}\frac{\partial \mathbf{C}_{\mathbf{rr}}}{\partial \theta}\mathbf{F}+\mathbf{F}\mathbf{C}_{\mathbf{rr}}\frac{\partial \mathbf{F}}{\partial\theta}$, and $\frac{\partial \mathbf{F}}{\partial\theta}$ and $\frac{\partial \mathbf{C}_{\mathbf{rr}}}{\partial \theta}$ are, respectively, given in the following expressions:
\begin{align}
	&\frac{\partial \mathbf{F}}{\partial\theta}=-\frac{1}{2}\sqrt{\frac{2}{\pi}}\operatorname{diag}\left(\mathbf{C}_{\mathbf{rr}}\right)^{-1}\operatorname{diag}\left(\frac{\partial \mathbf{C}_{\mathbf{rr}}}{\partial \theta}\right)\operatorname{diag}\left(\mathbf{C}_{\mathbf{rr}}\right)^{-\frac{1}{2}},\label{F theta derivative}\\
	&\frac{\partial \mathbf{C}_{\mathbf{rr}}}{\partial \theta}=\sigma_{\alpha}^2\frac{\partial \mathbf{A}_\theta}{\partial \theta}\mathbf{xx}^H\mathbf{A}_\theta^H+\sigma_\alpha^2\mathbf{A}_\theta \mathbf{xx}^H\frac{\partial \mathbf{A}_\theta^H}{\partial \theta},\label{Crr theta derivative}
\end{align}
where 
\begin{equation}
	\frac{\partial \mathbf{A}_\theta}{\partial \theta}=\mathbf{I}_L\otimes\left(\frac{\partial \mathbf{a}_r(\theta)}{\partial \theta}\mathbf{a}_t(\theta)^T+\mathbf{a}_r(\theta)\frac{\partial \mathbf{a}_t(\theta)^T}{\partial \theta}\right),\label{A theta derivative}
\end{equation}
with
$\frac{\partial \mathbf{a}_t(\theta)}{\partial \theta}=\mathbf{a}_t(\theta)\circ \left[0,-j\pi\cos(\theta),\dots,-j\pi(N_t-1)\cos(\theta)\right]^T$ and $\frac{\partial \mathbf{a}_r(\theta)}{\partial \theta}=\mathbf{a}_r(\theta)\circ \left[0,-j\pi\cos(\theta),\dots,-j\pi(N_r-1)\cos(\theta)\right]^T$.

\begin{remark}
Different from \cite{ref104} and \cite{ref202}, we employ the Bussgang decomposition to analyze the CRB for the parameter estimation in quantized ISAC systems. Furthermore, although \cite{ref49} derives a one-bit CRB for radar systems based on the Q-function, its highly nonlinear nature prohibits its application in waveform design. Alternatively, by combining the Bussgang-based covariance matrix and the worst-case Gaussian assumption, our proposed CRB yields a tractable structure which enables the subsequent efficient ISAC waveform optimization.
\end{remark}

\begin{remark}
It should also be pointed out that the CRB can not accurately characterize the actual estimator MSE in the low-SNR regime \cite{ref201}. However, compared to tighter global bounds (e.g., the Ziv-Zakai bound \cite{ref201}) that are analytically intractable, the CRB offers a closed-form surrogate and thus facilitates ISAC waveform design under one-bit quantization. Deriving a mathematically tractable and also tight performance bound for one-bit target parameter estimation in the low-SNR regime is worthwhile investigating in our future works.
\end{remark}

\subsection{Derivation of One-Bit CRB for Extended Targets}\label{subsection III B}
We continue by deriving the one-bit CRB metric for ET, where the target response matrix can be described as $\mathbf{A}_{\boldsymbol{\eta}}=\sum_{j=1}^J\alpha_j\mathbf{a}_r(\theta_j)\mathbf{a}_t(\theta_j)^T$, with $\boldsymbol{\eta}\triangleq\left[\alpha_1,\dots,\alpha_J,\theta_1,\dots,\theta_J\right]^T$ and $J$ being the number of scatterers of ET. Similar to the PT case detailed in Section \ref{subsection III A}, $\alpha_j$, $\mathbf{a}_t(\theta_j)$, and $\mathbf{a}_r(\theta_j)$ represent the reflection coefficient, the transmit array response, and the receive array response associated with the $j$-th scatterer, respectively. For ET, we still focus on estimating the DOAs $\theta_1,\dots,\theta_J$ and treat $\alpha_1,\dots,\alpha_J$ as nuisance parameters. By invoking Lemma \ref{lemma alpha}, the above nuisance parameters can be treated as Gaussian distributed and therefore the target response matrix $\mathbf{A}_{\boldsymbol{\eta}}$ follows a complex Gaussian distribution, i.e., satisfying $\operatorname{vec}\left(\mathbf{A}_{\boldsymbol{\eta}}\right)\sim\mathcal{CN}\left(\mathbf{0}_{N_rN_t},\mathbf{C}_{\mathbf{aa}}\right)$ with a known prior covariance $\mathbf{C}_{\mathbf{aa}}$. Moreover, instead of estimating $\theta_1,\dots,\theta_J$ directly, we adopt the methodology in \cite{ref24,ref25} by first estimating the entire target response matrix $\mathbf{A}_{\boldsymbol{\eta}}$, and thus the desired DOAs can be extracted from the estimate matrix via spectral estimation techniques, such as the one-bit MUSIC algorithm discussed in\cite{ref204}.

To estimate $\mathbf{A}_{\boldsymbol{\eta}}$, by recalling \eqref{echo signal r} and \eqref{linearized signal z}, we first express the Bussgang-based linearized signal model as
\begin{equation}
	\mathbf{z}=\mathbf{F}\tilde{\mathbf{X}}\mathbf{a}_{\boldsymbol{\eta}}+\tilde{\mathbf{v}},\label{linearized observation z}
\end{equation}
where we define $\tilde{\mathbf{X}}\triangleq\mathbf{X}^T\otimes \mathbf{I}_{N_r}$, and $\tilde{\mathbf{v}}\triangleq\mathbf{Fv}+\mathbf{q}$ is the effective noise uncorrelated with $\mathbf{a}_{\boldsymbol{\eta}}$. Furthermore, $\tilde{\mathbf{v}}$ has zero mean and its covariance matrix $\mathbf{C}_{\tilde{\mathbf{v}}\tilde{\mathbf{v}}}$, by noting that $\mathbf{a}_{\boldsymbol{\eta}}\sim\mathcal{CN}\left(\mathbf{0}_{N_rN_t},\mathbf{C}_{\mathbf{aa}}\right)$ and exploiting $\hat{\mathbf{C}}_{\mathbf{zz}}$ in \eqref{approximated covariance Czz}, can be approximated as
\begin{equation}
\begin{aligned}
\mathbf{C}_{\tilde{\mathbf{v}}\tilde{\mathbf{v}}}&\approx	\hat{\mathbf{C}}_{\mathbf{zz}}-\mathbf{F}\tilde{\mathbf{X}}\mathbf{C}_{\mathbf{aa}}\tilde{\mathbf{X}}^H\mathbf{F}\\
&=\mathbf{F}\mathbf{C}_{\mathbf{rr}}\mathbf{F}+\left(1-\frac{2}{\pi}\right)\mathbf{I}_{N_rL}-\mathbf{F}\tilde{\mathbf{X}}\mathbf{C}_{\mathbf{aa}}\tilde{\mathbf{X}}^H\mathbf{F}\\
&=\sigma_v^2\mathbf{FF}+\left(1-\frac{2}{\pi}\right)\mathbf{I}_{N_rL},\label{covariance matrix Cvv}
\end{aligned}
\end{equation}
where the last equality is due to $\mathbf{C}_{\mathbf{rr}}=\tilde{\mathbf{X}}\mathbf{C}_{\mathbf{aa}}\tilde{\mathbf{X}}^H+\sigma_v^2\mathbf{I}_{N_rL}$, as can be derived from \eqref{echo signal r}. Again we assume that $\tilde{\mathbf{v}}\sim\mathcal{CN}\left(\mathbf{0}_{N_rL},\mathbf{C}_{\tilde{\mathbf{v}}\tilde{\mathbf{v}}}\right)$, thereby leading to the worst-case (i.e., largest) CRB for estimating $\mathbf{a}_{\boldsymbol{\eta}}$ based on \eqref{linearized observation z}, as detailed in \cite{ref202}. Accordingly, the derived one-bit CRB for ET can be formulated as \cite[Chapter 15]{ref201}
\begin{equation}\label{extended targets CRB}
	\begin{aligned}
		\text{CRB}_{\mathbf{a}_{\boldsymbol{\eta}}}&=\operatorname{tr}\left(\left(\mathbf{C}_{\mathbf{aa}}^{-1}+\tilde{\mathbf{X}}^H\mathbf{F}\mathbf{C}_{\tilde{\mathbf{v}}\tilde{\mathbf{v}}}^{-1}\mathbf{F}\tilde{\mathbf{X}}\right)^{-1}\right).
	\end{aligned}
\end{equation}
Notably, since the derivation of $\text{CRB}_{\mathbf{a}_{\boldsymbol{\eta}}}$ incorporates the prior knowledge of $\mathbf{a}_{\boldsymbol{\eta}}$ (i.e., the covariance matrix $\mathbf{C}_{\mathbf{aa}}$), it formally constitutes a posterior or Bayesian CRB \cite{ref3,ref26}. Furthermore, this bound is analytically equivalent to the mean-squared error (MSE) achieved by the linear minimum mean-squared error (LMMSE) estimator that estimates $\mathbf{a}_{\boldsymbol{\eta}}$ under the linearized model in \eqref{linearized observation z}, which has been studied for quantized MIMO systems in \cite{ref44} to analyze the channel estimation performance.
\begin{remark}
The tractable CRB expressions derived in \eqref{point-like targets CRB} and \eqref{extended targets CRB} inherently rely on the linear approximation $\arcsin(x) \approx x,\ x \ne 1$ applied to the exact covariance matrix $\mathbf{C}_{\mathbf{zz}}$ in \eqref{covariance matrix Czz}. While adopting this approximation is favorable to guarantee the tractability of the subsequent waveform optimization, its theoretical validity should be further discussed. First, for massive MIMO systems, the off-diagonal elements of $\mathbf{C}_{\mathbf{zz}}$ in \eqref{covariance matrix Czz} are typically much smaller than 1 (i.e., the diagonal elements), thereby satisfying the condition for this approximation, see also \cite{ref44}. Moreover, to numerically validate the accuracy of this approximation, we evaluate the normalized $\text{MSE} = \frac{1}{N_rL} \Vert{} \mathbf{C}_{\mathbf{zz}} - \hat{\mathbf{C}}_{\mathbf{zz}} \Vert{}_F^2$ via simulations. As illustrated in Fig. \ref{fig.PTET_MSE}, the approximation error remains extremely small across a wide range of SNRs, thereby guaranteeing the accuracy of the derived CRB expressions.
\end{remark}

\begin{figure}[t]
	\centering
	\includegraphics[scale=0.6]{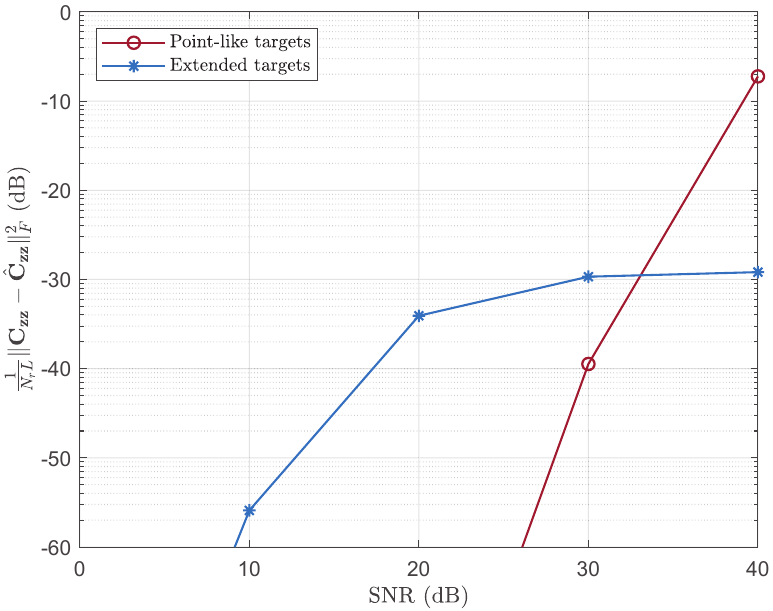}
	\caption{Normalized MSE for verifying the linear approximation $\arcsin(x) \approx x, \ x \ne 1$ adopted in $\hat{\mathbf{C}}_{\mathbf{zz}}$ ($N_t=N_r=16$ and $L=20$).}
	\label{fig.PTET_MSE}
\end{figure}

\subsection{Proposed One-Bit Estimation Methods}\label{subsection III C}
We now develop practical one-bit estimation methods for both PT and ET, which are based on the quantized observations in \eqref{quantized signal z} and approach the CRBs derived in \eqref{point-like targets CRB} and \eqref{extended targets CRB}, respectively. First, for the PT scenario, we propose a maximum likelihood estimator (MLE) for the DOA estimation. Specifically, by leveraging the assumption that $\mathbf{z}\sim\mathcal{CN}\left(\mathbf{0}_{N_rL},{\mathbf{C}}_{\mathbf{zz}}\right)$\footnote{Unlike the arcsin-approximated $\hat{\mathbf{C}}_{\mathbf{zz}}$ in \eqref{approximated covariance Czz} used for tractable CRB derivation, the exact $\mathbf{C}_{\mathbf{zz}}$ in \eqref{covariance matrix Czz} is adopted here to achieve better sensing accuracy.} detailed in Section \ref{subsection III A}, the corresponding log-likelihood function can be given by
\begin{equation}
	L(\mathbf{z};\theta)=-\mathbf{z}^H{\mathbf{C}}_{\mathbf{zz}}^{-1}\mathbf{z}-\ln\det(\pi{\mathbf{C}}_{\mathbf{zz}}),
\end{equation}
where ${\mathbf{C}}_{\mathbf{zz}}$ is defined in \eqref{covariance matrix Czz}. Thus, the one-bit MLE of $\theta$ can be obtained by maximizing $	L(\mathbf{z};\theta)$ with respect to $\theta$, which yields
\begin{equation}\label{MLE theta}
	\hat{\theta}_{\text{MLE}}=\mathop{\text{argmin}}\limits_{\theta\in[-\pi/2,\pi/2]}\ \mathbf{z}^H{\mathbf{C}}_{\mathbf{zz}}^{-1}\mathbf{z}+\ln\det({\mathbf{C}}_{\mathbf{zz}}), 
\end{equation}
where we ignore the constant term.
The optimal solution to \eqref{MLE theta} can be readily found via a one-dimensional search.

For the ET scenario, we are to estimate the target response matrix $\mathbf{A}_{\boldsymbol{\eta}}$ (i.e., $\mathbf{a}_{\boldsymbol{\eta}}$), as detailed in Section \ref{subsection III B}. Since the Bussgang-based signal $\mathbf{z}$ in \eqref{linearized observation z} essentially constitutes a linear observation of $\mathbf{a}_{\boldsymbol{\eta}}$, the LMMSE estimator can be utilized to obtain an estimate of $\mathbf{a}_{\boldsymbol{\eta}}$, which gives
\begin{equation}\label{LMMSE a}
	\begin{aligned}
		\hat{\mathbf{a}}_{\boldsymbol{\eta},\text{LMMSE}}=\mathbf{C}_{\mathbf{az}}\mathbf{C}_{\mathbf{zz}}^{-1}\mathbf{z},
	\end{aligned}
\end{equation}
with $\mathbf{C}_{\mathbf{az}}=\mathbf{C}_{\mathbf{aa}}\tilde{\mathbf{X}}^H\mathbf{F}$ and $\mathbf{C}_{\mathbf{zz}}$ defined in \eqref{covariance matrix Czz}.

It should be pointed out that the above Bussgang-based LMMSE (BLMMSE) estimator incurs a marginal performance degradation compared to the optimal minimum MSE (MMSE) estimator under one-bit quantization \cite{ref205}. However, unlike the analytically intractable MMSE approach, the BLMMSE estimator is inherently tied to the closed-form performance metric $\text{CRB}_{\mathbf{a}_{\boldsymbol{\eta}}}$ in \eqref{extended targets CRB}, which is of paramount importance, as it highly facilitates the efficient waveform design in the subsequent sections. Finally, the practical effectiveness of the proposed estimators in \eqref{MLE theta} and \eqref{LMMSE a} will be numerically validated in Section \ref{section VI}.

\begin{remark}
For the PT scenario, the observation is highly nonlinear with respect to the DOA parameter $\theta$. Therefore, the CRB serves as a fundamental lower bound, which is asymptotically approached by the MSE of the proposed nonlinear one-bit MLE. While for the ET scenario, since the target response parameter $\mathbf{a}_{\boldsymbol{\eta}}$ is linearly related to the linearized observation, the derived Bayesian CRB mathematically coincides with the MSE of the BLMMSE estimator.
\end{remark}

\section{ISAC Waveform Optimization for Point-Like Targets}\label{section IV}
In this section, we first formulate the bi-criterion\footnote{In ISAC systems that involve both communication and sensing performance metrics, the transceiver design problems are commonly formulated by either optimizing a weighted sum of both performance metrics \cite{ref27,ref48} or optimizing one metric while imposing a QoS constraint on the other\cite{ref21,ref22,ref23,ref24,ref25}.} ISAC waveform optimization problem for PT, which turns out to be a very challenging nonconvex problem due to the nonlinear CRB objective and the coupled variables within the SEP constraints. To address these issues, we develop an effective ADMM \cite{ref301} based algorithm, wherein the constructed subproblems in each iteration are solved by exploiting MM techniques and structured convexity.

\subsection{ISAC Problem Formulation}\label{subsection IV A}
Without loss of generality, we assume that $\varepsilon_k=\varepsilon$ and $\gamma_k=\gamma$ for $k\in\mathcal{K}$ and recast the SEP constraints in \eqref{SEP metric in vectorized form} by an abstract form $g(\mathbf{x},\mathbf{d})\le c(\varepsilon,\gamma)$ to simplify the notation, where $\mathbf{x}=\operatorname{vec}(\mathbf{X})=\left[\mathbf{x}_1^T,\dots,\mathbf{x}_L^T\right]^T$, $g(\cdot)$ is linear with respect to $\mathbf{x}$ and $\mathbf{d}$, and $c(\cdot)$ is  dependent on $\varepsilon$ and $\gamma$. Thus, by minimizing the one-bit CRB objective in \eqref{point-like targets CRB} and further imposing the above SEP constraint and a total power constraint, we formulate the bi-criterion ISAC waveform optimization problem for PT as
\begin{subequations}\label{ISAC problem PT DAC}
\begin{align}
\mathop{\text{maximize}}\limits_{\mathbf{x},\mathbf{d}}&\quad \operatorname{tr}\left(\hat{\mathbf{C}}_{\mathbf{zz}}^{-1}(\mathbf{x})\frac{\partial \hat{\mathbf{C}}_{\mathbf{zz}}(\mathbf{x})}{\partial \theta}\hat{\mathbf{C}}_{\mathbf{zz}}^{-1}(\mathbf{x})\frac{\partial \hat{\mathbf{C}}_{\mathbf{zz}}(\mathbf{x})}{\partial \theta}\right)\label{ISAC problem PT DAC I}\\
\text{subject to}
&\quad g(\mathbf{x},\mathbf{d})\le c(\varepsilon,\gamma) \label{ISAC problem PT DAC II}\\
&\quad \|\mathbf{x}\|_2^2\le P,\label{ISAC problem PT DAC III}
\end{align}
\end{subequations}
where $P$ is the total power budget at the ISAC transmitter. As can be seen, problem \eqref{ISAC problem PT DAC} is challenging to tackle due to the highly nonconvex objective in \eqref{ISAC problem PT DAC I} and the coupled variables $\mathbf{x}$ and $\mathbf{d}$ in \eqref{ISAC problem PT DAC II}.

\subsection{ADMM-Based Solution}\label{subsection IV B}
Directly tackling problem \eqref{ISAC problem PT DAC} is mathematically intractable due to the severe variable coupling between the nonconvex CRB objective \eqref{ISAC problem PT DAC I} and the complex SEP constraint \eqref{ISAC problem PT DAC II}. Here, we present an effective ADMM framework that utilizes a decomposition-coordination architecture to decouple the CRB minimization from the SEP constraint satisfaction, thereby making the inherent nonconvexity manageable. Specifically, by introducing an auxiliary variable $\mathbf{U}=\mathbf{HX}=\left[\mathbf{u}_1,\dots,\mathbf{u}_L\right]\in\mathbb{C}^{K\times L}$ and defining $\mathbf{u}=\operatorname{vec}(\mathbf{U})$, we first recast problem \eqref{ISAC problem PT DAC} by an equivalent form as follows:
\begin{subequations}\label{ISAC problem PT equality}
\begin{align}
\mathop{\text{minimize}}\limits_{\mathbf{x},\mathbf{u},\mathbf{d}}&\, f(\mathbf{x})\triangleq -\operatorname{tr}\left(\hat{\mathbf{C}}_{\mathbf{zz}}^{-1}(\mathbf{x})\frac{\partial \hat{\mathbf{C}}_{\mathbf{zz}}(\mathbf{x})}{\partial \theta}\hat{\mathbf{C}}_{\mathbf{zz}}^{-1}(\mathbf{x})\frac{\partial \hat{\mathbf{C}}_{\mathbf{zz}}(\mathbf{x})}{\partial \theta}\right)\label{ISAC problem PT equality I}\\
\text{subject to}
&\quad \|\mathbf{x}\|_2^2\le P \label{ISAC problem PT equality II}\\
&\quad \tilde{g}(\mathbf{u},\mathbf{d})\le c(\varepsilon,\gamma) \label{ISAC problem PT equality III}\\
&\quad \mathbf{u}=\tilde{\mathbf{H}}\mathbf{x},\label{ISAC problem PT equality IV}
\end{align}
\end{subequations}
where we define  $\tilde{g}(\mathbf{u},\mathbf{d})=g(\mathbf{x},\mathbf{d})$ and $\tilde{\mathbf{H}}=\mathbf{I}_L\otimes\mathbf{H}$. Then the augmented Lagrangian with respect to problem \eqref{ISAC problem PT equality} can be cast as \cite{ref302}
\begin{equation}
\begin{aligned}	
L_{\rho}(\mathbf{x},\mathbf{u},\boldsymbol{\lambda})
&= f(\mathbf{x})+2\Re\left(\bar{\boldsymbol{\lambda}}^H(\tilde{\mathbf{H}}\mathbf{x}-\mathbf{u})\right)+\rho\|\tilde{\mathbf{H}}\mathbf{x}-\mathbf{u}\|_2^2\\
&=f(\mathbf{x})+\rho \|\tilde{\mathbf{H}}\mathbf{x}-\mathbf{u}+\boldsymbol{\lambda}\|_2^2-\rho\|\boldsymbol{\lambda}\|_2^2,
\end{aligned}
\end{equation}
where $\bar{\boldsymbol{\lambda}}$ is the dual variable associated with the equality constraint in \eqref{ISAC problem PT equality IV} and we further denote by $\boldsymbol{\lambda}=\frac{1}{\rho}\bar{\boldsymbol{\lambda}}$ the scaled dual variable \cite{ref301}, and $\rho>0$ is the penalty parameter. According to the ADMM framework \cite{ref301}, we then arrive at the following subproblems in the $i$-th ADMM iteration:
\begin{equation}\label{ADMM subproblem I}
\begin{aligned}
\hspace{1cm}\mathbf{x}^{i+1}&=\mathop{\text{argmin}}\limits_{\mathbf{x}}\quad f(\mathbf{x})+\rho\Vert\tilde{\mathbf{H}}\mathbf{x}-\mathbf{u}^i+\boldsymbol{\lambda}^i\Vert_2^2\\
&\quad \text{subject to}\quad \eqref{ISAC problem PT equality II},
\end{aligned}
\end{equation}
\begin{equation}\label{ADMM subproblem II}
\begin{aligned}
\hspace{-1.58cm}\left\{\mathbf{u}^{i+1},\mathbf{d}^{i+1}\right\}&=\mathop{\text{argmin}}\limits_{\mathbf{u},\mathbf{d}}\quad \rho\Vert\tilde{\mathbf{H}}\mathbf{x}^{i+1}-\mathbf{u}+\boldsymbol{\lambda}^i\Vert_2^2\\ 
&\quad \text{subject to}\quad \eqref{ISAC problem PT equality III}, 
\end{aligned}
\end{equation}
\begin{equation}\label{ADMM dual ascend}
\begin{aligned}
\hspace{-1.1cm}\boldsymbol{\lambda}^{i+1}&= \boldsymbol{\lambda}^i + \left(\tilde{\mathbf{H}}\mathbf{x}^{i+1}-\mathbf{u}^{i+1}\right).
\end{aligned}
\end{equation}
Clearly, we see that the efficacy of the developed ADMM framework hinges on whether subproblems \eqref{ADMM subproblem I} and \eqref{ADMM subproblem II} can be solved effectively, which will be elaborated as follows.
\subsubsection{Solution to the $\mathbf{x}$-Subproblem}
By using \eqref{ISAC problem PT equality I}, we first reformulate subproblem \eqref{ADMM subproblem I} as
\begin{equation}\label{PLA subproblem I A}
\begin{aligned}
\mathop{\text{minimize}}\limits_{\mathbf{x}}&\quad -\operatorname{tr}\left(\hat{\mathbf{C}}_{\mathbf{zz}}^{-1}(\mathbf{x})\frac{\partial \hat{\mathbf{C}}_{\mathbf{zz}}(\mathbf{x})}{\partial \theta}\hat{\mathbf{C}}_{\mathbf{zz}}^{-1}(\mathbf{x})\frac{\partial \hat{\mathbf{C}}_{\mathbf{zz}}(\mathbf{x})}{\partial \theta}\right)\\&\quad+\rho\Vert\tilde{\mathbf{H}}\mathbf{x}-\mathbf{u}^i+\boldsymbol{\lambda}^i\Vert_2^2\\
\text{subject to}&\quad \eqref{ISAC problem PT equality II}.
\end{aligned}
\end{equation}
Applying $\operatorname{tr}\left(\mathbf{ABCD}\right)=\operatorname{vec}\left(\mathbf{A}^H\right)^H\left(\mathbf{D}^T\otimes \mathbf{B}\right)\operatorname{vec}\left(\mathbf{C}\right)$, we then recast problem \eqref{PLA subproblem I A} in an equivalent form as
\begin{equation}\label{PLA subproblem I B}
	\begin{aligned}
		\mathop{\text{minimize}}\limits_{\mathbf{x}}&\quad -\mathbf{p}(\mathbf{x})^H\mathbf{Q}(\mathbf{x})^{-1}\mathbf{p}(\mathbf{x})
 +\rho\Vert\tilde{\mathbf{H}}\mathbf{x}-\mathbf{u}^i+\boldsymbol{\lambda}^i\Vert_2^2
		\\
		\text{subject to}&\quad \eqref{ISAC problem PT equality II},
	\end{aligned}
\end{equation}
where we define $\mathbf{p}(\mathbf{x})=\frac{\partial \operatorname{vec}\left(\hat{\mathbf{C}}_{\mathbf{zz}}(\mathbf{x})\right)}{\partial \theta}$ and $\mathbf{Q}(\mathbf{x})=\hat{\mathbf{C}}_{\mathbf{zz}}^{T}(\mathbf{x})\otimes \hat{\mathbf{C}}_{\mathbf{zz}}(\mathbf{x})$. Note that the notation $(\mathbf{x})$ is dropped for simplicity in the rest of this section. As can be seen, the nonconvex objective makes problem \eqref{PLA subproblem I B} quite challenging. To address this issue, we first employ the MM technique \cite{ref303} to construct a surrogate form of problem \eqref{PLA subproblem I B}. More specifically, by noting that $-\mathbf{p}^H\mathbf{Q}^{-1}\mathbf{p}$ is jointly concave in $\mathbf{p}$ and $\mathbf{Q}$ \cite{ref304}, an upper bound of the objective of problem \eqref{PLA subproblem I B} can be derived via its first-order Taylor approximation \cite{ref303}, which gives rise to 
\begin{equation}\label{PT subproblem lower bound}
\begin{aligned}
&\quad -\mathbf{p}^H\mathbf{Q}^{-1}\mathbf{p}+\rho\Vert\tilde{\mathbf{H}}\mathbf{x}-\mathbf{u}^i+\boldsymbol{\lambda}^i\Vert_2^2\\
&\le -2\Re\left\{\mathbf{p}_t^H\mathbf{Q}_t^{-1}\mathbf{p}\right\}+\operatorname{tr}\left(\mathbf{Q}_t^{-1}\mathbf{p}_t\mathbf{p}_t^H\mathbf{Q}_t^{-1}\mathbf{Q}\right)\\
&\quad +\rho\Vert\tilde{\mathbf{H}}\mathbf{x}-\mathbf{u}^i+\boldsymbol{\lambda}^i\Vert_2^2+c_t,
\end{aligned}
\end{equation}
where the subscript $(\cdot)_t$ denotes the value at the $t$-th iteration, i.e., $\mathbf{p}_t=\mathbf{p}(\mathbf{x}_t), \mathbf{Q}_t=\mathbf{Q}(\mathbf{x}_t)$, and we denote by $c_t$ the constant term that does not affect the solution. Using the upper bound in \eqref{PT subproblem lower bound}, we therefore obtain a surrogate form of problem \eqref{PLA subproblem I B} expressed as
\begin{equation}\label{PLA subproblem I C}
	\begin{aligned}
		\mathop{\text{minimize}}\limits_{\mathbf{x}}&\quad \quad\  m(\mathbf{p},\mathbf{Q},\mathbf{x})\\
		&\quad \triangleq
		-2\Re\left\{\mathbf{p}_t^H\mathbf{Q}_t^{-1}\mathbf{p}\right\}+\operatorname{tr}\left(\mathbf{Q}_t^{-1}\mathbf{p}_t\mathbf{p}_t^H\mathbf{Q}_t^{-1}\mathbf{Q}\right)\\
		&\quad \quad +\rho\Vert\tilde{\mathbf{H}}\mathbf{x}-\mathbf{u}^i+\boldsymbol{\lambda}^i\Vert_2^2\\
		\text{subject to}&\quad \eqref{ISAC problem PT equality II}.
	\end{aligned}
\end{equation}
According to the MM framework, problem \eqref{PLA subproblem I B} can be solved by iteratively solving the surrogate problem in \eqref{PLA subproblem I C}. Although the objective of problem \eqref{PLA subproblem I C} is more tractable than that of problem \eqref{PLA subproblem I B}, it is still nonconvex in $\mathbf{x}$ and also difficult to tackle. Hence, we further adopt the projected gradient descent (PGD) method \cite{ref305} to solve problem \eqref{PLA subproblem I C}. To reduce the computational complexity, we only perform one PGD iteration as follows:
\begin{equation}\label{PGD update equation}
	\begin{aligned}
		\mathbf{x}_{t+1}=\mathcal{P}_C\left(\mathbf{x}_t-\mu\frac{\frac{\partial m(\mathbf{p},\mathbf{Q},\mathbf{x})}{\partial \mathbf{x}^*}}{\left\Vert\frac{\partial m(\mathbf{p},\mathbf{Q},\mathbf{x})}{\partial \mathbf{x}^*}\right\Vert_2}\right),
	\end{aligned}
\end{equation}
where $\mathcal{P}_C(\cdot)$ is the projection operator depending on the constraint in \eqref{ISAC problem PT equality II} and is thus given by
\begin{equation}\label{projection operator of box constraints}
\mathcal{P}_C(\mathbf{x})=\begin{cases}
	\mathbf{x}, &\|\mathbf{x}\|_2^2 \le P,\\
	\frac{\sqrt{P}}{\|\mathbf{x}\|_2}\mathbf{x}, &\|\mathbf{x}\|_2^2 > P,
\end{cases}	
\end{equation}
and $\mu$ is the step size given by the following backtracking line search method \cite{ref306}:
\begin{equation}
	\begin{aligned}
		&\quad\ m(\mathbf{p}_{t+1},\mathbf{Q}_{t+1},\mathbf{x}_{t+1})-m(\mathbf{p}_{t},\mathbf{Q}_{t},\mathbf{x}_t)\\
		&\le \frac{2\mu}{\left\Vert\frac{\partial m(\mathbf{p},\mathbf{Q},\mathbf{x})}{\partial \mathbf{x}^*}\right\Vert_2} \Re\left(\left(\frac{\partial m(\mathbf{p},\mathbf{Q},\mathbf{x})}{\partial \mathbf{x}^*}\right)^H\left(\mathbf{x}_{t+1}-\mathbf{x}_t\right)\right).\label{PGD step size}
	\end{aligned}
\end{equation}
Note that the above PGD method requires the expression for $\frac{\partial m(\mathbf{p},\mathbf{Q},\mathbf{x})}{\partial \mathbf{x}^*}$, which is derived in Appendix \ref{appendix derivative}.

\subsubsection{Solution to the $\left\{\mathbf{u},\mathbf{d}\right\}$-Subproblem}
Substituting the expressions of \eqref{SEP metric in vectorized form} into the constraint in \eqref{ISAC problem PT equality III}, subproblem \eqref{ADMM subproblem II} can be rewritten as
\begin{equation}\label{PT subproblem II A}
\begin{aligned}
\mathop{\text{minimize}}\limits_{\mathbf{u},\mathbf{d}}\quad& \left\Vert\mathbf{u}-\left(\tilde{\mathbf{H}}\mathbf{x}^{i+1}+\boldsymbol{\lambda}^i\right)\right\Vert_2^2	\\
\text{subject to}\quad &\\
-\mathbf{d}_R+\mathbf{b}_{R,l}&\le\Re\left(\mathbf{u}_l\right)-\mathbf{d}_R\circ \Re(\mathbf{s}_l) \le \mathbf{d}_R- \mathbf{a}_{R,l},\quad l\in\mathcal{L}\\
-\mathbf{d}_I+\mathbf{b}_{I,l}&\le\Im\left(\mathbf{u}_l\right)-\mathbf{d}_I\circ \Im(\mathbf{s}_l) \le \mathbf{d}_I- \mathbf{a}_{I,l},\quad l\in\mathcal{L}\\
\mathbf{d}_R&\ge \gamma \mathbf{1}_K,\quad \mathbf{d}_I\ge  \gamma \mathbf{1}_K.
\end{aligned}	
\end{equation}
Evidently problem \eqref{PT subproblem II A} is a solvable convex optimization problem. Although this problem can be solved via the interior point method (IPM) \cite{ref304}, it incurs a significant computational cost, thereby hindering its practical implementation for large-scale systems. In the remainder of this subsection, we develop a low-complexity problem-specific algorithm to find a high-quality solution to problem \eqref{PT subproblem II A}.

To start with, we reformulate problem \eqref{PT subproblem II A} as 
\begin{equation}\label{PT subproblem II B}
\begin{aligned}
\mathop{\text{minimize}}\limits_{\substack{\left\{\Re(\mathbf{u}_l)\right\}_{l=1}^L\\\left\{\Im(\mathbf{u}_l)\right\}_{l=1}^L\\\mathbf{d}_R,\mathbf{d}_I}}\quad &  \sum\limits_{l=1}^L\left(\Vert\Re(\mathbf{u}_l)-\Re(\tilde{\boldsymbol{\lambda}}_l^i)\Vert_2^2+\Vert\Im(\mathbf{u}_l)-\Im(\tilde{\boldsymbol{\lambda}}_l^i)\Vert_2^2\right)	\\
\text{subject to}\quad &\\
-\mathbf{d}_R+\mathbf{b}_{R,l}&\le\Re\left(\mathbf{u}_l\right)-\mathbf{d}_R\circ \Re(\mathbf{s}_l) \le \mathbf{d}_R- \mathbf{a}_{R,l},\quad l\in\mathcal{L}\\
-\mathbf{d}_I+\mathbf{b}_{I,l}&\le\Im\left(\mathbf{u}_l\right)-\mathbf{d}_I\circ \Im(\mathbf{s}_l) \le \mathbf{d}_I- \mathbf{a}_{I,l},\quad l\in\mathcal{L}\\
\mathbf{d}_R&\ge \gamma \mathbf{1}_K,\quad \mathbf{d}_I\ge  \gamma \mathbf{1}_K,
\end{aligned}	
\end{equation}
where we define $\left[\tilde{\boldsymbol{\lambda}}_1^i,\dots,\tilde{\boldsymbol{\lambda}}_L^i\right]=\operatorname{unvec}\left(\tilde{\mathbf{H}}\mathbf{x}^{i+1}+\boldsymbol{\lambda}^i\right)\in\mathbb{C}^{K\times L}$. Obviously, problem \eqref{PT subproblem II B} can be divided into $2K$ independent small-scale subproblems. More specifically, by defining $\tilde{\mathbf{u}}_l = \left[\Re(\mathbf{u}_l)^T, \Im(\mathbf{u}_l)^T\right]^T\in\mathbb{R}^{2K}$, $\boldsymbol{\chi}_l^i = \left[\Re(\tilde{\boldsymbol{\lambda}}_l^i)^T,  \Im(\tilde{\boldsymbol{\lambda}}_l^i)^T\right]^T\in\mathbb{R}^{2K}$, $\tilde{\mathbf{a}}_l = \left[\mathbf{a}_{R,l}^T,  \mathbf{a}_{I,l}^T\right]^T\in\mathbb{R}^{2K}$, $\tilde{\mathbf{b}}_l = \left[\mathbf{b}_{R,l}^T,  \mathbf{b}_{I,l}^T\right]^T\in\mathbb{R}^{2K}$, and $\tilde{\mathbf{s}}_l = \left[\Re(\mathbf{s}_l)^T,  \Im(\mathbf{s}_l)^T\right]^T\in\mathbb{R}^{2K}$, we then express the $k$-th subproblem as follows:
\begin{equation}\label{PT subproblem II scalar}
\begin{aligned}
\mathop{\text{minimize}} \limits_{\left\{\tilde{u}_{k,l}\right\}_{l=1}^L,d_k}&\  \sum\limits_{l=1}^L\left(\tilde{u}_{k,l}-\chi_{k,l}^i\right)^2\\
\text{subject to}&\ 
-{d}_k+\tilde{{b}}_{k,l}\le \tilde{u}_{k,l}-{d}_k \tilde{s}_{k,l} \le {d}_k- \tilde{{a}}_{k,l},\quad l\in\mathcal{L}\\
&\hspace{1.5cm} d_k\ge \gamma,
\end{aligned}
\end{equation}
with $\tilde{u}_{k,l}$, $\chi_{k,l}^i$, $\tilde{{a}}_{k,l}$, $\tilde{{b}}_{k,l}$, $\tilde{s}_{k,l}$, and $d_k$ being the $k$-th element of $\tilde{\mathbf{u}}_l$, $\boldsymbol{\chi}_l^i$, $\tilde{\mathbf{a}}_l$, $\tilde{\mathbf{b}}_l$, $\tilde{\mathbf{s}}_l$, and $\mathbf{d}$ (defined in Section \ref{subsection II A}). We can observe that problem \eqref{PT subproblem II scalar} is a convex quadratic problem with $2L+1$ linear inequality constraints. Motivated by the algorithmic framework in \cite[Algorithm 3]{ref102}, we develop a tailored method to tackle problem \eqref{PT subproblem II scalar}, with its optimal solution established in the following proposition.
\begin{proposition}\label{PT proposition I}
Denote the set of candidate solutions to $d_k$ by $\left\{d_{k,i}^\star\right\},i=1,\dots, \operatorname{card}(\mathcal{B})-1$, with $d_{k,i}^\star$ and $\mathcal{B}$ provided in Appendix \ref{appendix subproblem II solution}. Then, the optimal $d_k^\star$ is the one that minimizes $p(d_{k,i}^\star)=\sum_{l\in\Gamma_i} \left((\tilde{s}_{k,l}+1)d_{k,i}^\star-\tilde{a}_{k,l}-\chi_{k,l}^i\right)^2+\sum_{l\in\Omega_i} \left((\tilde{s}_{k,l}-1)d_{k,i}^\star+\tilde{b}_{k,l}-\chi_{k,l}^i\right)^2$, with $\Gamma_i$ and $\Omega_i$ defined in Appendix \ref{appendix subproblem II solution}. Furthermore, given $d_k^\star$, the optimal solution to $\tilde{u}_{k,l}$ for $l\in\mathcal{L}$ can be given by 
\begin{equation}\label{optimal solution to u}
\tilde{u}_{k,l}^\star=\left\{\begin{aligned}
&(\tilde{s}_{k,l}+1)d_k^\star-\tilde{a}_{k,l}  ,\quad && \chi_{k,l}^i> (\tilde{s}_{k,l}+1)d_k^\star-\tilde{a}_{k,l},\\
&(\tilde{s}_{k,l}-1)d_k^\star+\tilde{b}_{k,l}, \quad && \chi_{k,l}^i < (\tilde{s}_{k,l}-1)d_k^\star+\tilde{b}_{k,l},\\
&\chi_{k,l}^i,\quad && \text{otherwise}.
\end{aligned}\right.
\end{equation}
\end{proposition}

\begin{IEEEproof}
	See Appendix \ref{appendix subproblem II solution}.
\end{IEEEproof}

\begin{algorithm}[t]
	\caption{ADMM framework integrated with an MM-based PGD method (ADMM-MMPGD).}\label{summary of ADMM}
	\begin{algorithmic}[1]
		\Require{Feasible primal variables $\mathbf{x}^0$, $\mathbf{u}^0$, and $\mathbf{d}^0$, feasible dual variable $\boldsymbol{\lambda}^0$, total power budget $P$, and convergence accuracy $\epsilon$.}
		\State Set ADMM iteration index $i=0$.
		\Repeat
		\State Set MM iteration index $t=0$.
		\State $\mathbf{x}_t=\mathbf{x}^i$.
		\Repeat
		\State Obtain $\mathbf{x}_{t+1}$ via \eqref{PGD update equation}.
		\State $t \gets t+1$.
		\Until{convergence}.
		\State $\mathbf{x}^{i+1}=\mathbf{x}_{t}$.
		\State Obtain $\mathbf{u}^{i+1},\mathbf{d}^{i+1}$ via Proposition \ref{PT proposition I}.
		\State Obtain $\boldsymbol{\lambda}^{i+1}$ via \eqref{ADMM dual ascend}.
		\State $i \gets i+1$.
		\Until{convergence}.
		\Ensure{Optimized waveform $\mathbf{x}$ and decision variable $\mathbf{d}$.}
	\end{algorithmic}		
\end{algorithm}

To conclude this subsection, the proposed ADMM-based framework for solving problem \eqref{ISAC problem PT DAC}, which integrates the MM and PGD techniques, is summarized in Algorithm \ref{summary of ADMM} and referred to as ``ADMM-MMPGD". The computational cost of Algorithm \ref{summary of ADMM} is dominated by the PGD update in \eqref{PGD update equation} during each inner iteration. Specifically, the matrix inversion and multiplication operations involving $N_rL$-dimensional matrices incur a complexity of $\mathcal{O}\left((N_rL)^3\right)$. Thus, by denoting the number of outer and inner iterations of Algorithm \ref{summary of ADMM} as $I_{o}$ and $I_{i}$, respectively, the computational complexity of the ``ADMM-MMPGD" algorithm is $\mathcal{O}\left(I_{o}I_{i}\left(N_rL\right)^3\right)$.

\newcounter{MYtempeqncnt}
\begin{figure*}[ht]
	\normalsize
	\setcounter{MYtempeqncnt}{\value{equation}}
	\setcounter{equation}{46}
	\vspace*{4pt}
	\begin{equation}\label{extended CRB expansion}
		\begin{aligned}
			\text{CRB}_{\mathbf{a}_{\boldsymbol{\eta}}}&=\operatorname{tr}\left(\mathbf{C}_{\mathbf{aa}}\right)-\operatorname{tr}\left(\mathbf{C}_{\mathbf{aa}}\tilde{\mathbf{X}}^H\mathbf{F}\left(\mathbf{F}\tilde{\mathbf{X}}\mathbf{C}_{\mathbf{aa}}\tilde{\mathbf{X}}^H\mathbf{F}+\mathbf{C}_{\tilde{\mathbf{v}}\tilde{\mathbf{v}}}\right)^{-1}\mathbf{F}\tilde{\mathbf{X}}\mathbf{C}_{\mathbf{aa}}\right)\\
			&\approx\operatorname{tr}\left(\mathbf{C}_{\mathbf{aa}}\right)-\operatorname{tr}\left(\mathbf{C}_{\mathbf{aa}}\tilde{\mathbf{X}}^H\left(\tilde{\mathbf{X}}\mathbf{C}_{\mathbf{aa}}\tilde{\mathbf{X}}^H+\sigma_v^2\mathbf{I}_{N_rL}+\left(\frac{\pi}{2}-1\right)\operatorname{diag}\left(\tilde{\mathbf{X}}\mathbf{C}_{\mathbf{aa}}\tilde{\mathbf{X}}^H+\sigma_v^2\mathbf{I}_{N_rL}\right)\right)^{-1}\tilde{\mathbf{X}}\mathbf{C}_{\mathbf{aa}}\right).
		\end{aligned}
	\end{equation}
	\hrulefill
	\setcounter{equation}{\value{MYtempeqncnt}}
\end{figure*}

\begin{remark}\label{penalty strategy}
To improve the convergence performance of Algorithm \ref{summary of ADMM}, we enlarge the penalty parameter $\rho$ progressively, as detailed in \cite{ref301}. Specifically, upon completing the updates of the primal and dual variables ($\mathbf{x}$, $\mathbf{u}$, $\mathbf{d}$, and $\boldsymbol{\lambda}$) at each iteration, the penalty parameter is increased via $\rho\leftarrow c_\rho \rho$, where $c_\rho>1$ denotes a predefined scaling factor. Moreover, the dual variable is rescaled as $\boldsymbol{\lambda}\leftarrow \boldsymbol{\lambda}/c_\rho$ to maintain theoretical consistency \cite{ref301}. This adaptive procedure continues until convergence is achieved or $\rho$ reaches a predefined upper bound $\rho_{\max}$. Note that similar techniques have also been adopted in \cite{ref42,ref43,ref45} to boost algorithmic efficiency.
\end{remark}

\begin{remark}
It is worth noting that establishing a rigorous theoretical convergence proof for the proposed ADMM algorithm is mathematically intractable. This is mainly due to the severe nonconvexity of the parameter estimation problem under one-bit quantization. Nevertheless, we evaluate the residual convergence (indicating constraint satisfaction) and the objective convergence (indicating performance stability) \cite{ref301} via extensive numerical simulations in Section VI.
\end{remark}

\section{ISAC Waveform Optimization for Extended Targets}\label{section V}
In this section, we formulate the bi-criterion ISAC waveform design problem for the ET scenario and then obtain an ADMM-based solution to the resulting problem.

\subsection{ISAC Problem Formulation}
First, by invoking the matrix inverse lemma and approximations in \eqref{covariance matrix Cvv}, we rewrite the one-bit CRB metric for ET in \eqref{extended targets CRB} at the top of the next page. Furthermore, by defining $\mathbf{L}(\mathbf{x})\triangleq \tilde{\mathbf{X}}\mathbf{C}_{\mathbf{aa}}$ and $\mathbf{M}(\mathbf{x})\triangleq\tilde{\mathbf{X}}\mathbf{C}_{\mathbf{aa}}\tilde{\mathbf{X}}^H+\left(\frac{\pi}{2}-1\right)\operatorname{diag}\left(\tilde{\mathbf{X}}\mathbf{C}_{\mathbf{aa}}\tilde{\mathbf{X}}^H\right)+\frac{\pi\sigma_v^2}{2}\mathbf{I}_{N_rL}$, we obtain
\setcounter{equation}{47}
\begin{equation}\label{extended CRB objective}
	\begin{aligned}
		\text{CRB}_{\mathbf{a}_{\boldsymbol{\eta}}}=\operatorname{tr}\left(\mathbf{C}_{\mathbf{aa}}\right)-\operatorname{tr}\left(\mathbf{L}(\mathbf{x})^H\mathbf{M}^{-1}(\mathbf{x})\mathbf{L}(\mathbf{x})\right).
	\end{aligned}
\end{equation}
Then, by taking into account the above CRB objective and the constraints in \eqref{ISAC problem PT DAC II}, \eqref{ISAC problem PT DAC III}, we formulate the bi-criterion ISAC waveform optimization problem for ET as
\begin{subequations}\label{ISAC problem ET DAC}
	\begin{align}
		\mathop{\text{minimize}}\limits_{\mathbf{x},\mathbf{d}}&\quad h(\mathbf{x})\triangleq-\operatorname{tr}\left(\mathbf{L}(\mathbf{x})^H\mathbf{M}^{-1}(\mathbf{x})\mathbf{L}(\mathbf{x})\right)\label{ISAC problem ET DAC I}\\
		\text{subject to}
		&\quad g(\mathbf{x},\mathbf{d})\le c(\varepsilon,\gamma)\label{ISAC problem ET DAC II}\\
		&\quad \|\mathbf{x}\|_2^2\le P,\label{ISAC problem ET DAC III}
	\end{align}
\end{subequations}
which is also a highly nonconvex problem similar to the PT case.

\subsection{ADMM-Based Solution}
Since problem \eqref{ISAC problem ET DAC} has a similar structure as problem \eqref{ISAC problem PT DAC}, we also employ the ADMM framework to solve problem \eqref{ISAC problem ET DAC}. To this end, we first construct the augmented Lagrangian associated with problem \eqref{ISAC problem ET DAC} as follows:
\begin{equation}
L_{\rho}(\mathbf{x},\mathbf{u},\boldsymbol{\lambda})
=h(\mathbf{x})+\rho \|\tilde{\mathbf{H}}\mathbf{x}-\mathbf{u}+\boldsymbol{\lambda}\|_2^2-\rho\|\boldsymbol{\lambda}\|_2^2,
\end{equation}
where $\mathbf{u}$, $\boldsymbol{\lambda}$, $\tilde{\mathbf{H}}$, and $\rho$ are, respectively, the auxiliary variable, the dual variable, the effective channel matrix, and the penalty parameter, as defined in Section \ref{subsection IV B}. Furthermore, the corresponding subproblems in the $i$-th ADMM iteration can be expressed as
\begin{equation}\label{ET ADMM subproblem I}
\begin{aligned}
\hspace{1cm}\mathbf{x}^{i+1}&= \mathop{\text{argmin}}\limits_{\mathbf{x}}\quad h(\mathbf{x})+\rho\Vert\tilde{\mathbf{H}}\mathbf{x}-\mathbf{u}^i+\boldsymbol{\lambda}^i\Vert_2^2\\
&\quad \text{subject to}\quad \eqref{ISAC problem ET DAC III},
\end{aligned}
\end{equation}
\begin{equation}\label{ET ADMM subproblem II}
\begin{aligned}
\hspace{-2cm}\left\{\mathbf{u}^{i+1},\mathbf{d}^{i+1}\right\}&=\mathop{\text{argmin}}\limits_{\mathbf{u},\mathbf{d}}\quad \rho\Vert\tilde{\mathbf{H}}\mathbf{x}^{i+1}-\mathbf{u}+\boldsymbol{\lambda}^i\Vert_2^2\\ 
&\quad \text{subject to}\quad \eqref{ISAC problem ET DAC II},
\end{aligned}
\end{equation}
\begin{equation}\label{Et ADMM dual ascend}
\begin{aligned}
\hspace{-1.1cm}\boldsymbol{\lambda}^{i+1}&= \boldsymbol{\lambda}^i + \left(\tilde{\mathbf{H}}\mathbf{x}^{i+1}-\mathbf{u}^{i+1}\right).
\end{aligned}
\end{equation}
In the remainder of this subsection, we focus on tackling subproblem \eqref{ET ADMM subproblem I}, while subproblem \eqref{ET ADMM subproblem II} shares the same form as subproblem \eqref{ADMM subproblem II} and has been addressed in Section \ref{subsection IV B}.

We begin by reformulating subproblem \eqref{ET ADMM subproblem I} as follows:
\begin{equation}\label{ET subproblem I A}
\begin{aligned}
\mathop{\text{minimize}}\limits_{\mathbf{x}}&\ -\operatorname{tr}\left(\mathbf{L}(\mathbf{x})^H\mathbf{M}(\mathbf{x})^{-1}\mathbf{L}(\mathbf{x})\right)+\rho\|\tilde{\mathbf{H}}\mathbf{x}-\mathbf{u}^i+\boldsymbol{\lambda}^i\|_2^2\\
\text{subject to}&\ \eqref{ISAC problem ET DAC III}.
\end{aligned}
\end{equation}
The primary difficulty in solving problem \eqref{ET subproblem I A} arises from the nonconvexity of its objective function with respect to $\mathbf{x}$. To address this problem, we develop an MM-based iterative algorithm, which guarantees monotonic convergence to a stationary point of problem \eqref{ET subproblem I A}. The constructed surrogate problem for each iteration is provided in the subsequent theorem.
\begin{theorem}\label{ET theorem I}
The MM-based surrogate problem corresponding to problem \eqref{ET subproblem I A} is formulated as follows:
\begin{equation}\label{ET subproblem I B}
\begin{aligned}
	\mathop{\text{minimize}}\limits_{\mathbf{x}}&\quad \left(\lambda_{\max}\left(\bar{\mathbf{M}}_t\right)+\rho\lambda_{\max}\left(\tilde{\mathbf{H}}^H\tilde{\mathbf{H}}\right)\right)\|\mathbf{x}\|_2^2\\&\quad-2\Re\left(\mathbf{m}_t^H\mathbf{x}\right)\\
	\text{subject to}&\quad \eqref{ISAC problem ET DAC III},	
\end{aligned}
\end{equation}
where $\bar{\mathbf{M}}_t$ and $\mathbf{m}_t$ are defined in \eqref{auxiliary variable Mt} and \eqref{expression of mt}, respectively.
\end{theorem} 

\begin{IEEEproof}
	See Appendix \ref{appendix ET problem solution}.
\end{IEEEproof}

As can be seen, problem \eqref{ET subproblem I B} is clearly a convex quadratic problem, whose optimal solution can be given by
\begin{equation}
	\mathbf{x}_{t+1}=\mathcal{P}_C\left(\frac{\mathbf{m}_t}{\lambda_{\max}\left(\bar{\mathbf{M}}_t\right)+\rho\lambda_{\max}\left(\tilde{\mathbf{H}}^H\tilde{\mathbf{H}}\right)}\right),\label{MM update equation}
\end{equation}
where $\mathcal{P}_C(\cdot)$ is the projector operator given in \eqref{projection operator of box constraints}. Hence, problem \eqref{ISAC problem ET DAC} can be effectively solved by adopting the proposed Algorithm \ref{summary of ADMM}, where the update of $\mathbf{x}_{t+1}$ is obtained via \eqref{MM update equation} instead of \eqref{PGD update equation}. To distinguish this approach from the ``ADMM-MMPGD" algorithm developed in Section \ref{subsection IV B}, we term this modified ADMM framework exploiting the MM technique to yield a closed-form solution as ``ADMM-MMCF". Furthermore, evaluating $\mathbf{x}_{t+1}$ in \eqref{MM update equation} incurs a complexity of $\mathcal{O}\left((N_rL)^3+N_r^3N_tL^2\right)$ per iteration, dominated by high-dimensional matrix inversions and multiplications. Therefore, the computational complexity of the proposed ADMM-MMCF algorithm is given by $\mathcal{O}\left(I_{o}I_{i}\left((N_rL)^3+N_r^3N_tL^2\right)\right)$.

\section{Simulation Results}\label{section VI}
This section presents simulation results to evaluate the tightness of the proposed one-bit CRB and the effectiveness of the developed ISAC waveform design in both the PT and ET scenarios. Unless otherwise specified, the system parameters are configured with $N_t=N_r=16$ transmit and receive antennas, $K=4$ downlink UEs, and a block length of $L=20$. In particular, for PT, the DOA to be estimated is set to $\theta = 30^\circ$. 
The reflection coefficient $\alpha$ is generated by normalizing a random variable drawn from $\mathcal{CN}(0,1)$.
Moreover, the ET target response matrix $\mathbf{A}_{\boldsymbol{\eta}}$ to be estimated is characterized by the Kronecker model, given by
\begin{equation}
\mathbf{A}_{\boldsymbol{\eta}}=\mathbf{\Phi}_R^{1/2}\mathbf{A}_{\text{i.i.d.}} \mathbf{\Phi}_T^{1/2},
\end{equation}
where the entries of $\mathbf{A}_{\text{i.i.d.}}$ are independent and identically distributed (i.i.d.) complex Gaussian variables with zero mean and unit variance. The receive and transmit correlation matrices, $\mathbf{\Phi}_R$ and $\mathbf{\Phi}_T$, respectively, are generated following the exponential correlation model \cite{ref501} with a correlation coefficient of $0.5$. Additionally, the sensing and communication signal-to-noise ratios (SNRs) are defined as $\frac{P}{\sigma_v^2}$ and $\frac{P}{\sigma_w^2}$, respectively. All simulation results are averaged over 100 independent Monte Carlo realizations.

\begin{figure}[t]
	\centering
	\subfloat[Residual $\Vert\tilde{\mathbf{H}}\mathbf{x}-\mathbf{u}\Vert_2^2$]{%
		\includegraphics[scale=0.31]{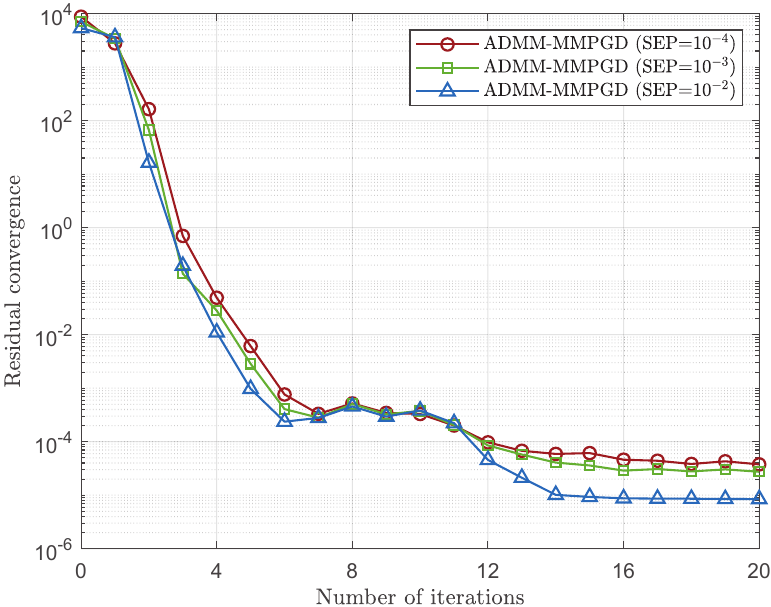}
		\label{fig.PT_residual_convergnce}
	}
	\hfill
	\subfloat[Objective $\text{CRB}_\theta$]{%
		\includegraphics[scale=0.31]{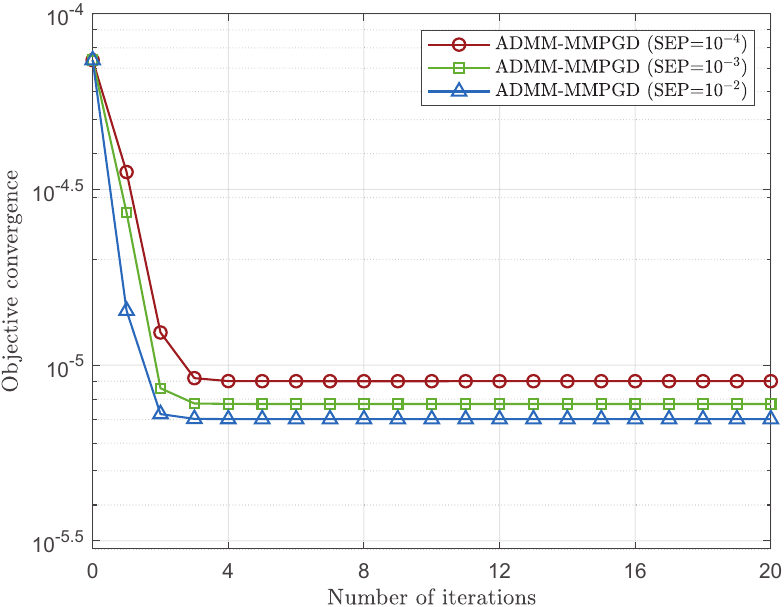}	\label{fig.PT_objective_convergence}
	}
	\caption{Residual and objective convergence of the ADMM-MMPGD algorithm under different SEPs ($N_t=N_r=16, K=4, L=20$, and $\text{SNR}=30\,\text{dB}$).}
	\label{fig.PT_convergence}
\end{figure}

\begin{figure}[t]
	\centering
	\subfloat[Residual $\Vert\tilde{\mathbf{H}}\mathbf{x}-\mathbf{u}\Vert_2^2$]{
		\includegraphics[scale=0.31]{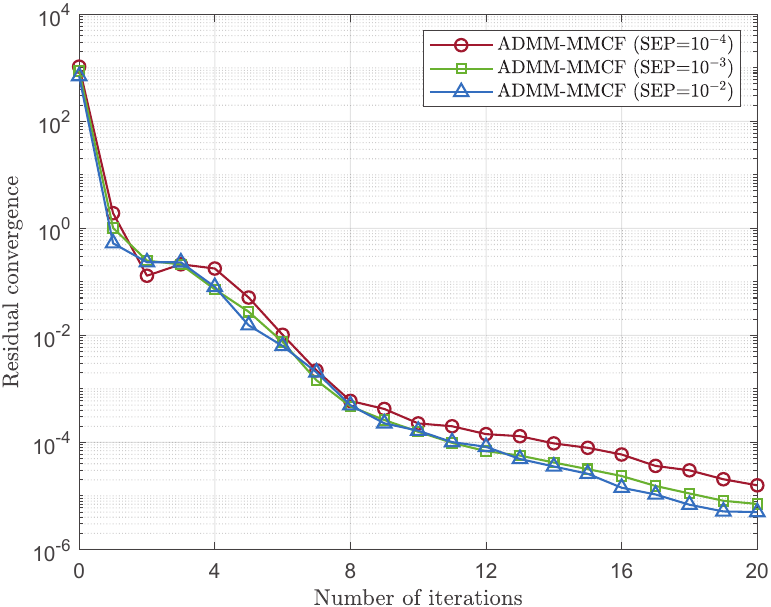}
		\label{fig.ET_residual_convergnce}
	}
	\hfill
	\subfloat[Objective $\text{CRB}_{\mathbf{a}_{\boldsymbol{\eta}}}/\operatorname{tr}(\mathbf{C}_{\mathbf{aa}})$]{
		\includegraphics[scale=0.31]{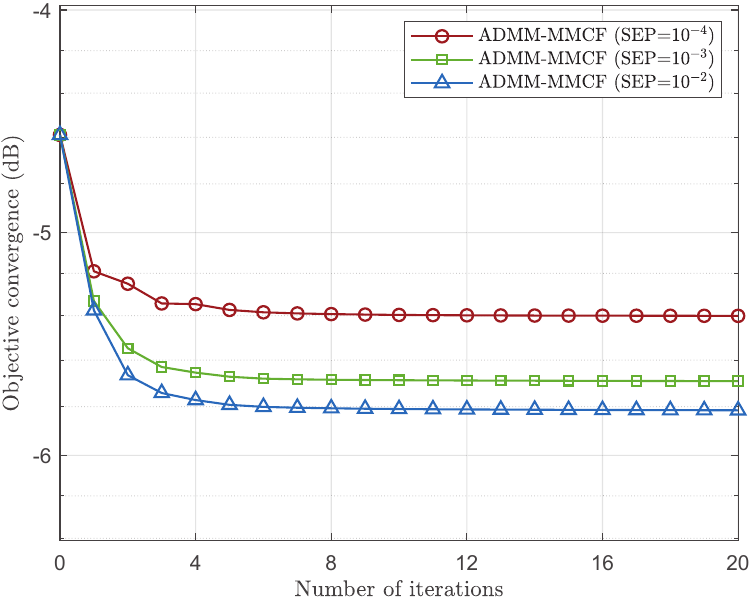}	\label{fig.ET_objective_convergence}
	}
	\caption{Residual and objective convergence of the ADMM-MMCF algorithm under different SEPs ($N_t=N_r=16, K=4, L=20$, and $\text{SNR}=30\,\text{dB}$).}
	\label{fig.ET_convergence}
\end{figure}

\subsection{Convergence Analysis}
We begin by verifying the average convergence of the developed ``ADMM-MMPGD" and ``ADMM-MMCF" algorithms in Fig.~\ref{fig.PT_convergence} and Fig.~\ref{fig.ET_convergence}, respectively, which illustrate the residual, defined as $\|\tilde{\mathbf{H}}\mathbf{x}-\mathbf{u}\|_2^2$, and the objectives $\text{CRB}_{\theta}$ and $\text{CRB}_{\mathbf{a}_{\boldsymbol{\eta}}}$, specified in \eqref{point-like targets CRB} and \eqref{extended targets CRB}, respectively, as a function of the number of iterations under different SEP requirements. 
Specifically, Fig.~\ref{fig.PT_residual_convergnce} shows that the residual $\|\tilde{\mathbf{H}}\mathbf{x}-\mathbf{u}\|_2^2$ decreases as the ``ADMM-MMPGD" algorithm iterates until it reaches final convergence, and meanwhile Fig.~\ref{fig.PT_objective_convergence} shows that the objective $\text{CRB}_{\theta}$ rapidly converges within a few iterations and appears inconsistent with its residual convergence process in Fig.~\ref{fig.PT_residual_convergnce}. This is because the ``ADMM-MMPGD" algorithm initially minimizes the augmented objective, i.e., $\text{CRB}_{\theta}$ plus the penalty $\rho\|\tilde{\mathbf{H}}\mathbf{x}-\mathbf{u}\|_2^2$, with a small penalty parameter $\rho$ to escape poor local minima, and then imposes a large $\rho$ to minimize the residual $\|\tilde{\mathbf{H}}\mathbf{x}-\mathbf{u}\|_2^2$ to enforce residual convergence, thereby yielding a trivial impact on the CRB objective optimization.
Similar observations can also be found in  Fig.~\ref{fig.ET_residual_convergnce} and Fig.~\ref{fig.ET_objective_convergence} for the ``ADMM-MMCF" algorithm, where Fig.~\ref{fig.ET_residual_convergnce} shows that the residual $\|\tilde{\mathbf{H}}\mathbf{x}-\mathbf{u}\|_2^2$ continuously reduces and falls below $10^{-4}$ rapidly and Fig.~\ref{fig.ET_objective_convergence} shows that the CRB objective converges rapidly and remains constant as the optimization process becomes dominated by the residual penalty when $\rho$ enlarges. 

\begin{figure}[t]
	\centering
	\includegraphics[scale=0.6]{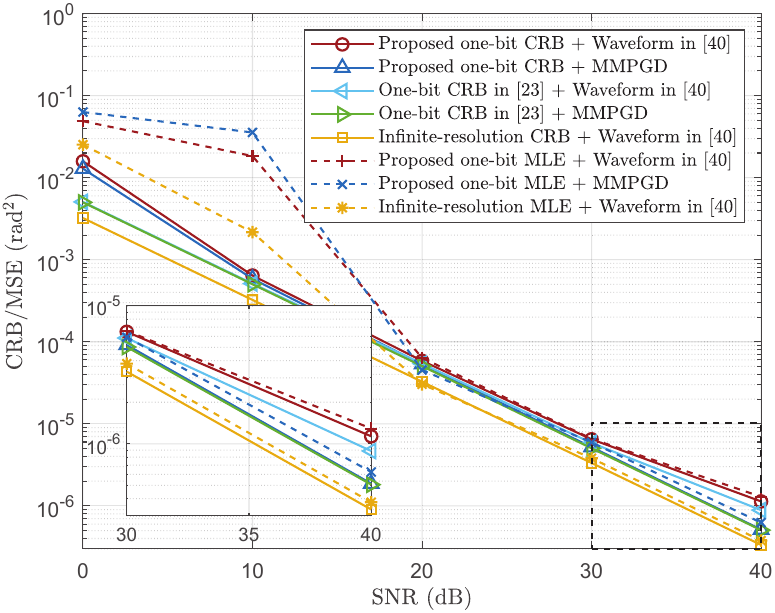}
	\caption{CRB and MSE performance comparisons between the proposed and benchmark waveforms for PT ($N_t=N_r=16$ and $L=20$).}
	\label{fig.PT_estimation}
\end{figure}

\subsection{Sensing Performance Comparison}
We now compare the CRB and MSE performance of the proposed waveforms against benchmark schemes in both the PT and ET scenarios. Since the communication performance is not examined in this simulation, only the inner iterations of the ``ADMM-MMPGD" and the ``ADMM-MMCF" algorithms are invoked, and thus these purely estimation-oriented waveform designs are referred to as the ``MMPGD" and the ``MMCF" algorithms, respectively. For the PT benchmark, we adopt the waveform design from \cite{ref502}, which is obtained by minimizing the infinite-resolution CRB via the off-the-shelf tool CVX \cite{ref503}. For the ET scenario, the benchmark is generated by adapting the ``MMCF" algorithm to solve a quantization-unaware MSE minimization problem, hereafter denoted as ``QU-MMCF". Furthermore, the MSE for PT and the normalized MSE for ET are defined as $\mathbb{E}\left\{\|\hat{\theta}-\theta\|_2^2\right\}$ and $\mathbb{E}\left\{\|\hat{\mathbf{a}}_{\boldsymbol{\eta}}-\mathbf{a}_{\boldsymbol{\eta}}\|_2^2\right\}/\operatorname{tr}(\mathbf{C}_{\mathbf{aa}})$, respectively, with $\hat{\theta}$, $\theta$, $\hat{\mathbf{a}}_{\boldsymbol{\eta}}$, and $\mathbf{a}_{\boldsymbol{\eta}}$ denoting the estimated DOA, the true DOA, the estimated target response, and the true target response, respectively.

Fig.~\ref{fig.PT_estimation} compares the CRB and MSE performance between the proposed and benchmark waveforms for the PT scenario. To facilitate a thorough evaluation, we also include the conventional infinite-resolution CRB derived in \cite[Eq. (63)]{ref504} and the Q-function-based one-bit CRB from \cite[Eq. (15)]{ref49}. As can be seen, for both the proposed and the Q-function-based one-bit CRBs, the ``MMPGD" algorithm outperforms the benchmark waveform in \cite{ref502}, especially at high SNRs. This is because the waveform in \cite{ref502} is designed under the assumption of infinite-resolution quantization, thus suffering a non-negligible performance loss when deployed with one-bit ADCs. 
Furthermore, the Q-function-based one-bit CRB in \cite{ref49} is shown to be lower than the proposed one, which is due to the fact that our proposed Bussgang-based one-bit CRB is derived under the Gaussian assumption, yielding a worst-case lower bound. Nevertheless, unlike the non-analytical Q-function-based bound, the proposed one-bit CRB is mathematically tractable, thereby facilitating efficient waveform optimization. 
In addition, compared to the benchmark waveform with infinite-resolution quantization, a performance degradation of about $-2\,\text{dB}\approx \frac{2}{\pi}$ can be observed for the ``MMPGD" based waveform with one-bit CRB, which quantifies the estimation performance loss due to the use of one-bit ADCs and is also consistent with the analysis result in \cite{ref49}. Lastly, we also see that the MSE of the developed one-bit MLE approaches the proposed one-bit CRB at high SNRs, which resembles the infinite-resolution quantization case and indicates the effectiveness of the proposed one-bit sensing method. Finally, note that the simulated MSE is slightly lower than the CRB at 20 dB. This is due to the approximation of $\mathbf{C}_{\mathbf{zz}}$ in the CRB.

\begin{figure}[t]
	\centering
	\includegraphics[scale=0.6]{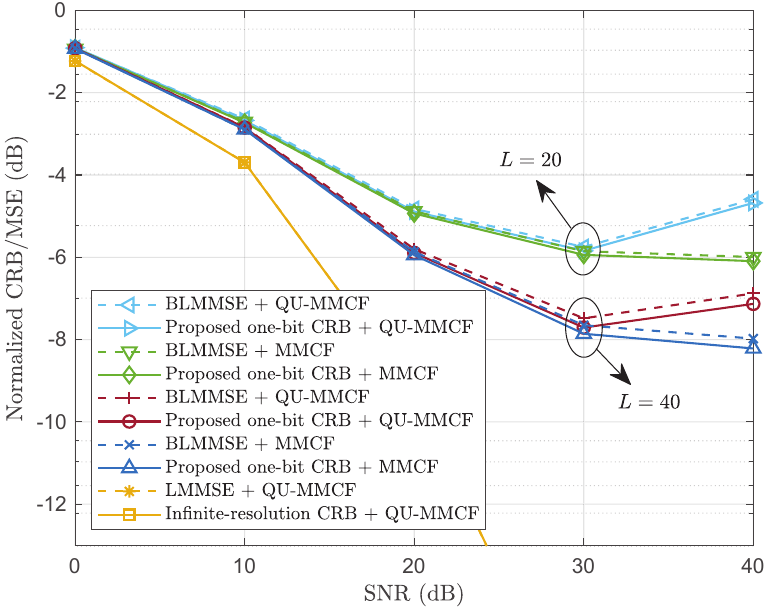}
	\caption{CRB and MSE performance comparisons between the proposed and benchmark waveforms for ET ($N_t=N_r=16$).}
	\label{fig.ET_estimation}
\end{figure}

Fig.~\ref{fig.ET_estimation} depicts the normalized CRB, defined as $\text{CRB}_{\mathbf{a}_{\boldsymbol{\eta}}}/\operatorname{tr}(\mathbf{C}_{\mathbf{aa}})$, alongside the normalized MSE performance of the proposed and benchmark waveforms for the ET scenario. First, we can observe that the MSE of the BLMMSE estimator in \eqref{LMMSE a} closely aligns with the theoretical CRB in \eqref{extended targets CRB} across various SNRs and block lengths for both waveforms, which verifies both the tightness of the proposed one-bit CRB for ET and the validity of the BLMMSE estimator. 
Furthermore, the performance of the ``QU-MMCF" scheme evidently degrades at SNRs exceeding $30\,\text{dB}$, whereas the ``MMCF" algorithm demonstrates consistent improvement. This phenomenon is due to the fact that, for the quantization-unaware ``QU-MMCF" scheme, a moderate level of noise power is actually beneficial for parameter estimation, which is known as the \emph{stochastic resonance} effect \cite{ref205}. In contrast, the proposed one-bit CRB explicitly incorporates the quantization effect, enabling the ``MMCF" algorithm to effectively suppress quantization noise and achieve noticeable performance gains.
Additionally, in contrast to the PT results in Fig.~\ref{fig.PT_estimation}, the CRB and MSE curves of the ``MMCF" algorithm in Fig.~\ref{fig.ET_estimation} saturate at SNRs above $40\,\text{dB}$, and the performance gap between the quantized and unquantized schemes also becomes substantially large. This saturation arises because the ET scenario requires estimating a high-dimensional matrix with $N_rN_t$ parameters, as opposed to a single scalar DOA parameter in the PT case. Consequently, the ET estimation problem is intrinsically more sensitive to the quantization precision of the received signals.

\begin{figure}[t]
	\centering
	\subfloat[PT case]{%
		\includegraphics[scale=0.31]{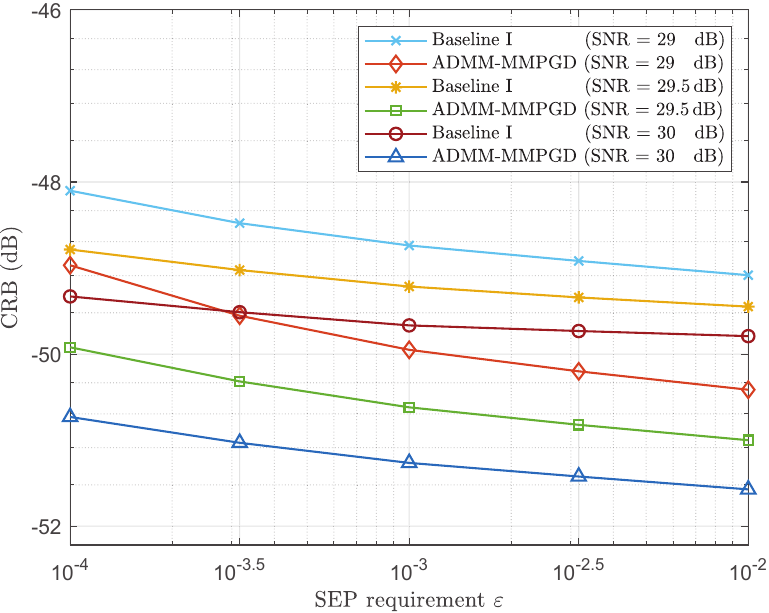}
		\label{fig.PT_ISAC}
	}
	\hfill
	\subfloat[ET case]{%
		\includegraphics[scale=0.31]{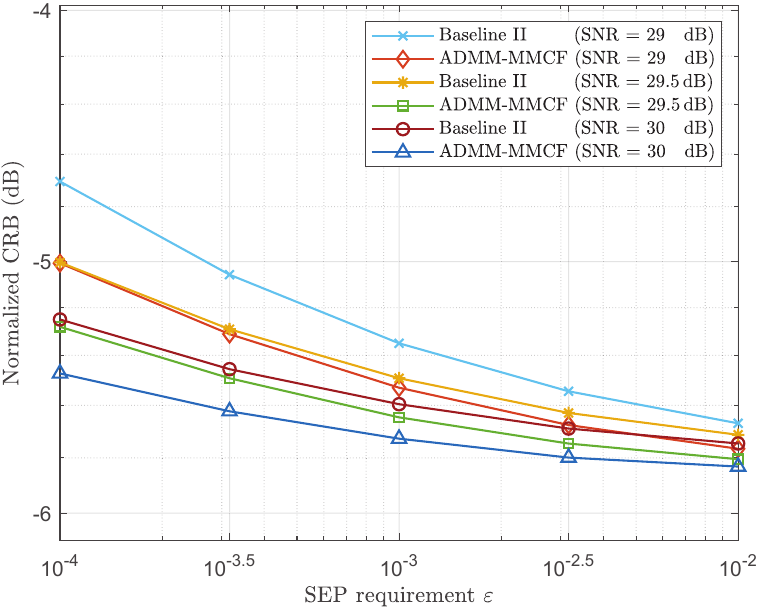}	\label{fig.ET_ISAC}
	}
	\caption{CRB versus the SEP requirement $\varepsilon$ for the proposed and benchmark schemes ($N_t=N_r=16$, $K=4$, and $L=20$).}
	\label{fig.ISAC_performance}
\end{figure}

\begin{table*}[t]
	\renewcommand{\arraystretch}{2}
	\caption{Computational Complexity and CPU Time Comparison}
	\label{table I}
	\centering
	\begin{tabular}{|c|c|c|c|}
		\hline 
		\multirow{2}{*}{Algorithm}& \multirow{2}{*}{Computational complexity}& Average CPU time (in s) & Average CPU time (in s) \\
		& & ($N_t=N_r=8$ and $L=10$) & ($N_t=N_r=16$ and $L=20$)\\
		\hline
		ADMM-MMPGD & $\mathcal{O}\left(I_{o,\text{ADMM-MMPGD}}I_{i,\text{ADMM-MMPGD}}\left(N_rL\right)^3\right)$& 2.79 & 52.67
		\\
		\hline
		Baseline I & $\mathcal{O}\left(I_{o,\text{Baseline I}}I_{i,\text{Baseline I}}\left(N_rL\right)^3\right)$ & 1.15 & 18.91\\
		\hline
		ADMM-MMCF&
		$\mathcal{O}\left(I_{o,\text{ADMM-MMCF}}I_{i,\text{ADMM-MMCF}}\left(\left(N_rL\right)^3+N_r^3N_tL^2\right)\right)$
		& 1.03 & 56.64\\
		\hline
		Baseline II & $\mathcal{O}\left(I_{o,\text{Baseline II}}I_{i,\text{Baseline II}}\left(\left(N_rL\right)^3+N_r^3N_tL^2\right)\right)$ & 1.21 & 58.39\\
		\hline
	\end{tabular}
\end{table*}

\subsection{ISAC Performance Evaluation}
This subsection evaluates the ISAC performance of the proposed one-bit waveform designs in both the PT and ET scenarios. The $16$-QAM constellation is employed, and the entries of the downlink channel $\mathbf{H}$ are modeled as i.i.d. complex Gaussian variables with zero mean and unit variance. To initialize the ``ADMM-MMPGD" and ``ADMM-MMCF" algorithms, the transmit waveform $\mathbf{x}$ is randomly generated subject to the power constraint, the dual variable $\boldsymbol{\lambda}$ is initialized as an all-zero vector, and the auxiliary variable $\mathbf{u}$ is computed via \eqref{optimal solution to u} with the initial decision variable $\mathbf{d}=\gamma\mathbf{1}_{2K}$. The convergence accuracy is set to $\epsilon = 10^{-4}$. Moreover, the penalty parameters detailed in Remark \ref{penalty strategy} are configured as $(\rho,c_\rho,\rho_{\max})=(10^2,3,10^{12})$ for the PT scenario and $(\rho,c_\rho,\rho_{\max})=(1,1.1,10)$ for the ET scenario.
To benchmark the proposed methods, we derive two baseline schemes by ignoring the one-bit quantization effect at the sensing receiver. Specifically, for the PT scenario, the one-bit CRB objective in \eqref{ISAC problem PT DAC I} is replaced by its infinite-resolution counterpart, i.e., $\operatorname{tr}\left(\mathbf{C}_{\mathbf{rr}}^{-1}\frac{\partial \mathbf{C}_{\mathbf{rr}}}{\partial \theta}\mathbf{C}_{\mathbf{rr}}^{-1}\frac{\partial \mathbf{C}_{\mathbf{rr}}}{\partial \theta}\right)$, with $\mathbf{C}_{\mathbf{rr}}$ defined in \eqref{covariance matrix Crr}, and the resulting ISAC waveform design is denoted as ``Baseline I". For the ET scenario, the quantization-unaware ``QU-MMCF" scheme (previously evaluated in Fig.~\ref{fig.ET_estimation}) is integrated into the proposed ADMM framework, serving as ``Baseline II". The corresponding ISAC performance results for the PT and ET cases are plotted in Fig.~\ref{fig.PT_ISAC} and Fig.~\ref{fig.ET_ISAC}, respectively.

From Fig.~\ref{fig.PT_ISAC}, we first observe that the performance trade-off between the CRB and the SEP is clearly evident for both the ``ADMM-MMPGD" and ``Baseline I" schemes, where an increase in the SEP requirement $\varepsilon$ leads to an improved one-bit CRB performance. Furthermore, the above trade-off relationship becomes more pronounced at lower SNRs. This behavior is because, at low SNRs, the developed ADMM framework tends to impose a larger penalty term on the CRB objective to satisfy the SEP constraint, which inevitably compromises the sensing performance. In addition, for a given SEP requirement $\varepsilon$, the ``ADMM-MMPGD" algorithm achieves a substantial CRB reduction compared to the ``Baseline I" method, demonstrating the advantage of the proposed ISAC waveform optimization.

Fig.~\ref{fig.ET_ISAC} illustrates the ISAC performance of the ``ADMM-MMCF" and ``Baseline II" algorithms. Consistent with the PT results in Fig.~\ref{fig.PT_ISAC}, a similar trade-off between the CRB and SEP performance can be observed for both schemes. Moreover, the ``ADMM-MMCF" algorithm also exhibits a lower one-bit CRB compared to the ``Baseline II" method, which validates the validity of the proposed one-bit CRB as a waveform design metric in one-bit quantized scenarios.

Finally, Table~\ref{table I} includes the computational complexity and CPU time of both the proposed and benchmark schemes evaluated in this subsection. The simulation is performed in MATLAB on a desktop with Intel Core i7-10700 CPU and 32 GB RAM. As shown in Table~\ref{table I}, the proposed and benchmark algorithms incur a relatively high computational cost, especially when the number of antennas and the block length are large, which is a cost of achieving the superior ISAC performance demonstrated above. Developing lower-complexity algorithms while maintaining the performance advantage will be an important topic of our future work.




\begin{figure}[t]
	\centering
	\subfloat[PT]{
		\includegraphics[scale=0.31]{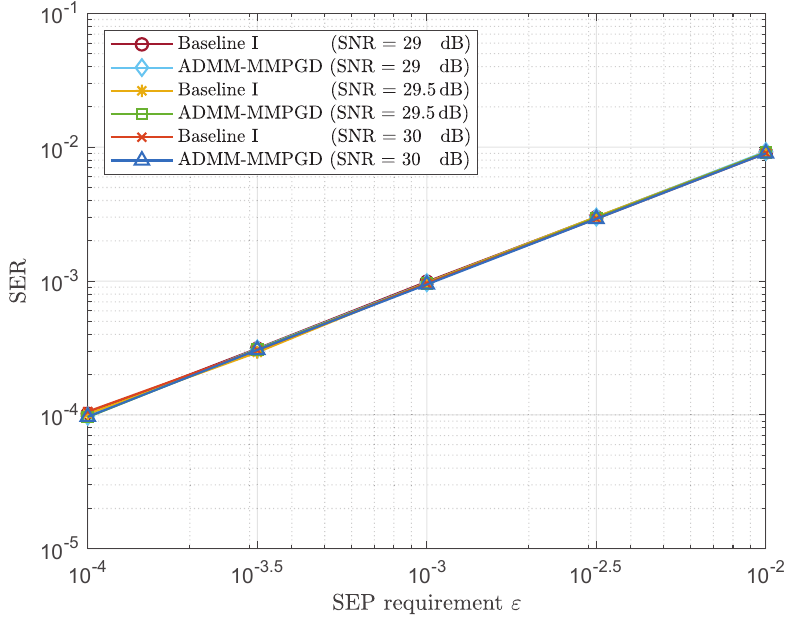}
		\label{fig.PT_SER_verification}
	}
	\hfill
	\subfloat[ET]{
		\includegraphics[scale=0.31]{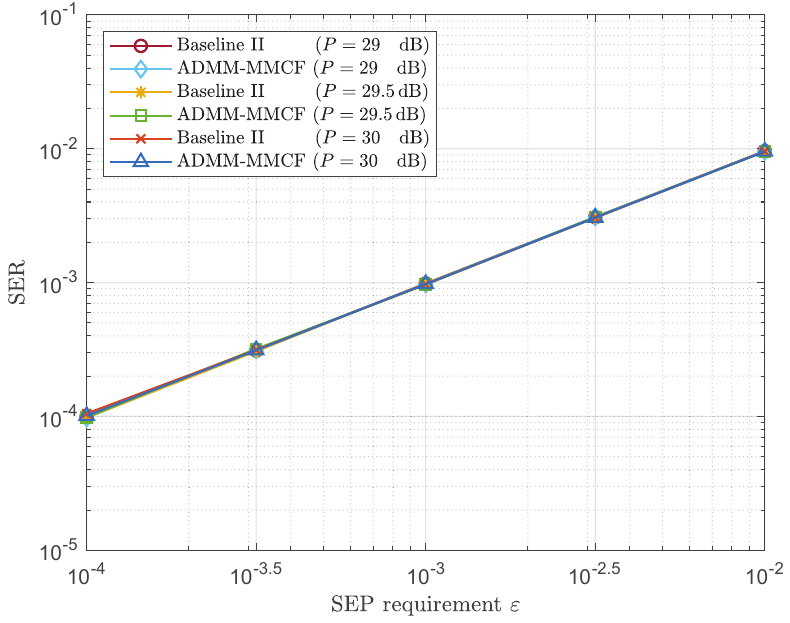}	\label{fig.ET_SER_verification}
	}
	\caption{Achievable SER versus SEP requirement $\varepsilon$ for both ADMM-MMPGD and ADMM-MMCF algorithms ($N_t=N_r=16, K=4$, and $L=20$).}
	\label{fig.SER_verification}
\end{figure}

\subsection{SER Performance Verification}\label{section VI D}
Lastly, we validate the relationship between the achievable SER performance and the theoretical SEP requirement, where the corresponding simulation results for PT and ET scenarios are shown in Fig.~\ref{fig.PT_SER_verification} and Fig.~\ref{fig.ET_SER_verification}, respectively.
As can be seen, the results in both figures show that the achievable SER is nearly identical to the SEP requirement for both proposed and benchmark schemes, which verifies that the proposed waveform optimization method can readily obtain a desirable SER performance by tuning the SEP requirement $\epsilon$.

\section{Conclusion}\label{section VII}
We have performed an in-depth study on the parameter estimation and waveform optimization for MIMO ISAC systems equipped with one-bit ADCs. By leveraging the Bussgang theorem and the worst-case Gaussian assumption, we derived novel one-bit CRB metrics for both the PT and ET scenarios, which can be approached by the developed one-bit estimation methods. Building upon the proposed CRBs and the SEP criterion, we investigated a novel ISAC waveform design problem, for which we developed an effective ADMM framework incorporated with the MM technique to find a high-quality solution. Numerical results verify the tightness of the proposed one-bit CRBs and the superiority of our optimized waveforms over existing benchmarks. Finally, the proposed ISAC design was shown to facilitate a flexible trade-off between sensing and communication performance.

\appendices

\section{Derivation of the Gradient}\label{appendix derivative}
The objective function of problem \eqref{PLA subproblem I C} can be cast as
\begin{equation}
	\begin{aligned}	
		m(\mathbf{p},\mathbf{Q},\mathbf{x})&=-2\Re\left\{\mathbf{p}_t^H\mathbf{Q}_t^{-1}\mathbf{p}\right\}+\operatorname{tr}\left(\mathbf{Q}_t^{-1}\mathbf{p}_t\mathbf{p}_t^H\mathbf{Q}_t^{-1}\mathbf{Q}\right)\\
		&\quad +\rho\|\tilde{\mathbf{H}}\mathbf{x}-\mathbf{u}^i+\boldsymbol{\lambda}^i\|_2^2,\\
		&\triangleq m_1+m_2+m_3+m_4,
	\end{aligned}
\end{equation}
where we define $m_1=-\mathbf{p}_t^H\mathbf{Q}_t^{-1}\mathbf{p}$, $m_2=-\mathbf{p}_t^T\mathbf{Q}_t^{-T}\mathbf{p}^*$, $m_3=\operatorname{tr}\left(\mathbf{Q}_t^{-1}\mathbf{p}_t\mathbf{p}_t^H\mathbf{Q}_t^{-1}\mathbf{Q}\right)$, and $m_4=\rho\|\tilde{\mathbf{H}}\mathbf{x}-\mathbf{u}^i+\boldsymbol{\lambda}^i\|_2^2$, for simplicity of derivation. Thus, the conjugate gradient $\frac{\partial m(\mathbf{p},\mathbf{Q},\mathbf{x})}{\partial \mathbf{x}^*}$ can be expressed as
\begin{equation}
	\begin{aligned}
		\frac{\partial m(\mathbf{p},\mathbf{Q},\mathbf{x})}{\partial \mathbf{x}^*}=\left(\frac{\partial m_1}{\partial \mathbf{x}^T}+\frac{\partial m_2}{\partial \mathbf{x}^T}+\frac{\partial m_3}{\partial \mathbf{x}^T}+\frac{\partial m_4}{\partial \mathbf{x}^T}\right)^H.\label{objective gradient sum}
	\end{aligned}
\end{equation}
In the following, we present the expressions for $\frac{\partial m_1}{\partial \mathbf{x}^T}$, $\frac{\partial m_3}{\partial \mathbf{x}^T}$, and $\frac{\partial m_4}{\partial \mathbf{x}^T}$, respectively, while the expression for $\frac{\partial m_2}{\partial \mathbf{x}^T}$ has a similar form as that of $\frac{\partial m_1}{\partial \mathbf{x}^T}$ due to the fact that $m_2$ is the conjugate of $m_1$, and is thus omitted here for brevity. To simplify the notation, we first define auxiliary variables $\mathbf{P}_t=\operatorname{unvec}\left(\mathbf{Q}_t^{-1}\mathbf{p}_t\right)\in\mathbb{C}^{N_rL\times N_rL}$, $\mathbf{K}_1=\mathbf{C}_{\mathbf{rr}}\frac{\partial \mathbf{F}}{\partial \theta}\mathbf{P}_t+\mathbf{P}_t\mathbf{F}\frac{\partial \mathbf{C}_{\mathbf{rr}}}{\partial \theta}+\mathbf{P}_t\frac{\partial \mathbf{F}}{\partial \theta}\mathbf{C}_{\mathbf{rr}}+\frac{\partial \mathbf{C}_{\mathbf{rr}}}{\partial \theta}\mathbf{F}\mathbf{P}_t$, and $\mathbf{K}_2=\mathbf{C}_{\mathbf{rr}}\mathbf{F}\mathbf{P}_t+\mathbf{P}_t\mathbf{F}\mathbf{C}_{\mathbf{rr}}$, with  $\mathbf{C}_{\mathbf{rr}}$, $\frac{\partial \mathbf{F}}{\partial \theta}$, and $\frac{\partial \mathbf{C}_{\mathbf{rr}}}{\partial \theta}$ given in \eqref{covariance matrix Crr}, \eqref{F theta derivative}, and \eqref{Crr theta derivative}, respectively. Therefore, we can express $\frac{\partial m_1}{\partial \mathbf{x}^T}$ as
\begin{equation}
	\frac{\partial m_1}{\partial \mathbf{x}^T}=\frac{\partial m_{1,1}}{\partial \mathbf{x}^T}+\frac{\partial m_{1,2}}{\partial \mathbf{x}^T}+\frac{\partial m_{1,3}}{\partial \mathbf{x}^T}+\frac{\partial m_{1,4}}{\partial \mathbf{x}^T}+\frac{\partial m_{1,5}}{\partial \mathbf{x}^T}+\frac{\partial m_{1,6}}{\partial \mathbf{x}^T},
\end{equation}
where 
\begin{align}
	\frac{\partial m_{1,1}}{\partial \mathbf{x}^T}&=-\sigma_\alpha^2\mathbf{p}_t^H\mathbf{Q}_t^{-1}\left(\frac{\partial \mathbf{F}}{\partial \theta}\otimes\mathbf{F}+\mathbf{F}\otimes \frac{\partial \mathbf{F}}{\partial \theta}\right)\notag\\&\quad\times\left(\left(\mathbf{x}^H\mathbf{A}_\theta^H\right)^T\otimes\mathbf{A}_\theta\right),\\		
	\frac{\partial m_{1,2}}{\partial \mathbf{x}^T}&=-\sigma_\alpha^2\mathbf{p}_t^H\mathbf{Q}_t^{-1}\left(\mathbf{F}\otimes \mathbf{F}\right)\notag\\&\quad\times\left(\left(\mathbf{x}^H\mathbf{A}_\theta^H\right)^T\otimes\frac{\partial \mathbf{A}_\theta}{\partial \theta}
	+\left(\mathbf{x}^H\frac{\partial \mathbf{A}_\theta^H}{\partial \theta}\right)^T\otimes\mathbf{A}_\theta\right),\\
		\frac{\partial m_{1,3}}{\partial \mathbf{x}^T}&=\frac{\sigma_\alpha^2}{2}\sqrt{\frac{2}{\pi}}\mathbf{x}^H\mathbf{A}_\theta^H\operatorname{diag}\left(\mathbf{J}_2\mathbf{K}_1^H\mathbf{J}_1\right)\mathbf{A}_\theta,
\end{align}	
\begin{align}	
	\frac{\partial m_{1,4}}{\partial \mathbf{x}^T}&=-\frac{\sigma_\alpha^2}{2}\sqrt{\frac{2}{\pi}}\mathbf{x}^H\mathbf{A}_\theta^H\notag\\&\quad\times\operatorname{diag}\left(\mathbf{J}_1\operatorname{diag}\left(\frac{\partial \mathbf{C}_{\mathbf{rr}}}{\partial \theta}\right)\mathbf{J}_2\mathbf{K}_2^H
	\mathbf{J}_1\right)\mathbf{A}_\theta,\\		
	\frac{\partial m_{1,5}}{\partial \mathbf{x}^T}&=\frac{\sigma_\alpha^2}{2}\sqrt{\frac{2}{\pi}}\mathbf{x}^H\mathbf{A}_\theta^H\operatorname{diag}\left(\mathbf{J}_2\mathbf{K}_2^H
	\mathbf{J}_1\right)\frac{\partial \mathbf{A}_\theta}{\partial \theta}\notag
	\\ &\quad +
	\frac{\sigma_\alpha^2}{2}\sqrt{\frac{2}{\pi}}\mathbf{x}^H\frac{\partial \mathbf{A}_\theta^H}{\partial \theta}\operatorname{diag}\left(\mathbf{J}_2\mathbf{K}_2^H
	\mathbf{J}_1\right)\mathbf{A}_\theta,\\
	\frac{\partial m_{1,6}}{\partial \mathbf{x}^T}&=-\frac{\sigma_\alpha^2}{4}\sqrt{\frac{2}{\pi}}\mathbf{x}^H\mathbf{A}_\theta^H\notag\\&\quad\times\operatorname{diag}\left(\mathbf{J}_2\mathbf{K}_2^H
	\mathbf{J}_1\operatorname{diag}\left(\frac{\partial \mathbf{C}_{\mathbf{rr}}}{\partial \theta}\right)\mathbf{J}_1\right)\mathbf{A}_\theta,
\end{align}
with $\mathbf{A}_\theta$ and  $\frac{\partial \mathbf{A}_\theta}{\partial \theta}$ given in \eqref{simplified notation Atheta} and \eqref{A theta derivative}, respectively, and we further define $\mathbf{J}_1=\operatorname{diag}\left(\mathbf{C}_{\mathbf{rr}}\right)^{-1}$ and $\mathbf{J}_2=\operatorname{diag}\left(\mathbf{C}_{\mathbf{rr}}\right)^{-\frac{1}{2}}$, respectively, for notational convenience.

We continue by deriving the expression for $\frac{\partial m_3}{\partial \mathbf{x}^T}$. By defining $\mathbf{N}_t=-\left(\mathbf{P}_t^*\otimes \mathbf{P}_t\right)$, $\mathbf{k}=\mathbf{T}_{N_rL,N_rL}\mathbf{N}_t^T\operatorname{vec}\left(\mathbf{C}_{\mathbf{zz}}^*\right)+\mathbf{N}_t\mathbf{T}_{N_rL,N_rL}\operatorname{vec}\left(\mathbf{C}_{\mathbf{zz}}^*\right)$, and $\mathbf{K}_3=\operatorname{unvec}\left(\mathbf{k}\right)\in\mathbb{C}^{N_rL\times N_rL}$, we then express $\frac{\partial m_3}{\partial \mathbf{x}^T}$ as
\begin{equation}
	\begin{aligned}
		\frac{\partial m_3}{\partial \mathbf{x}^T}&=\frac{\sigma_\alpha^2}{2}\sqrt{\frac{2}{\pi}}\mathbf{x}^H\mathbf{A}_\theta^H\operatorname{diag}\left(\mathbf{J}_2\mathbf{C}_{\mathbf{rr}}\mathbf{F}\mathbf{K}_3^H\mathbf{J}_1\right)\mathbf{A}_\theta\\
		&\quad+\frac{\sigma_\alpha^2}{2}\sqrt{\frac{2}{\pi}}\mathbf{x}^H\mathbf{A}_\theta^H\operatorname{diag}\left(\mathbf{J}_2\mathbf{K}_3^H\mathbf{F}\mathbf{C}_{\mathbf{rr}}\mathbf{J}_1\right)\mathbf{A}_\theta\\
		&\quad -\sigma_\alpha^2\mathbf{x}^H\mathbf{A}_\theta^H\mathbf{F}\mathbf{K}_3^H\mathbf{F}\mathbf{A}_\theta.
	\end{aligned}
\end{equation}
Lastly, we obtain the expression for $\frac{\partial m_4}{\partial \mathbf{x}^T}$ as follows:
\begin{equation}
\begin{aligned}
\frac{\partial m_4}{\partial \mathbf{x}^T}&=\rho\left(\tilde{\mathbf{H}}\mathbf{x}-\mathbf{u}^i+\boldsymbol{\lambda}^i\right)^H\tilde{\mathbf{H}}.
\end{aligned}	
\end{equation}
Notably, the evaluation of the aforementioned derivations can be simplified by using $\left(\mathbf{C}^T\otimes \mathbf{A}\right)\operatorname{vec}(\mathbf{B})=\operatorname{vec}\left(\mathbf{ABC}\right)$, which effectively circumvents the complexity of high-dimensional matrix computations.

\section{Proof of Proposition \ref{PT proposition I}}\label{appendix subproblem II solution}
We start by reformulating problem \eqref{PT subproblem II scalar} as a scalar problem with respect to $d_k$ whose optimal solution can be readily achieved, and then the optimal solution to $\tilde{u}_{k,l},l\in\mathcal{L}$ is obtained in closed form given the solution $d_k^\star$. Specifically, for a given $d_k$, problem \eqref{PT subproblem II scalar} reduces to a quadratic program subject to box constraints \cite{ref304} whose optimal solution is
\begin{equation}\label{solution to ul}
\tilde{u}_{k,l}^\star=\left\{\begin{aligned}
&(\tilde{s}_{k,l}+1)d_k-\tilde{a}_{k,l}  ,\quad && \chi_{k,l}^i> (\tilde{s}_{k,l}+1)d_k-\tilde{a}_{k,l},\\
&(\tilde{s}_{k,l}-1)d_k+\tilde{b}_{k,l}, \quad && \chi_{k,l}^i < (\tilde{s}_{k,l}-1)d_k+\tilde{b}_{k,l},\\
&\chi_{k,l}^i,\quad && \text{otherwise}.
\end{aligned}\right.
\end{equation}
Then, by substituting the above solutions of $\tilde{u}_{k,l},l\in\mathcal{L}$  back into \eqref{PT subproblem II scalar}, we thus arrive at a one-dimensional convex problem with respect to $d_k$, which, however, has different forms, depending on different solutions of $\tilde{u}_{k,l}$ in \eqref{solution to ul}. To proceed, we first split the feasible set of $d_k$ into different subsets, the possible boundary points of which, by taking into account the inequalities in \eqref{solution to ul} and $d_k\ge \gamma$ in problem \eqref{PT subproblem II scalar}, can be given by $\gamma$, $\left\{\left.\frac{\chi_{k,l}^i+\tilde{a}_{k,l}}{1+\tilde{s}_{k,l}},\frac{\tilde{b}_{k,l}-\chi_{k,l}^i}{1-\tilde{s}_{k,l}}\right\vert \frac{\chi_{k,l}^i+\tilde{a}_{k,l}}{1+\tilde{s}_{k,l}}>\gamma, \frac{\tilde{b}_{k,l}-\chi_{k,l}^i}{1-\tilde{s}_{k,l}}> \gamma\right\}_{l=1}^L$, and $\infty$. Moreover, we sort the above boundary points into an ascending order and denote the set of the sorted boundary points by $\mathcal{B}\triangleq\left\{\tau_1,\dots,\tau_{\operatorname{card}(\mathcal{B})}\right\}$, thereby yielding $\operatorname{card}(\mathcal{B})-1$ feasible subsets, i.e., $\tau_{i}\le d_{k,i}\le \tau_{i+1}, i=1,\dots,\operatorname{card}(\mathcal{B})-1$, where we replace $d_k$ with $d_{k,i}$ to highlight the index of the subset. Hence, the problem associated with the $i$-th feasible subset can be cast as
\begin{equation}\label{d optimization problem}
	\begin{aligned}
		\mathop{\text{minimize}}\limits_{d_{k,i}}&\quad p\left(d_{k,i}\right)\triangleq \sum\limits_{l\in\Gamma_i} \left((\tilde{s}_{k,l}+1)d_{k,i}-\tilde{a}_{k,l}-\chi_{k,l}^i\right)^2\\&\hspace{1.8cm}  +\sum\limits_{l\in\Omega_i} \left((\tilde{s}_{k,l}-1)d_{k,i}+\tilde{b}_{k,l}-\chi_{k,l}^i\right)^2\\
		\text{subject to }&\quad \tau_i\le d_{k,i} \le \tau_{i+1},
	\end{aligned}
\end{equation}
where we define $\Gamma_i\triangleq \left\{l\left\vert \chi_{k,l}^i> (\tilde{s}_{k,l}+1)\tau_{i/ i+1}-\tilde{a}_{k,l}\right.\right\}$ and $\Omega_i\triangleq \left\{l\left\vert \chi_{k,l}^i< (\tilde{s}_{k,l}-1)\tau_{i/i+1}+\tilde{b}_{k,l}\right.\right\}$. Obviously, problem \eqref{d optimization problem} admits the following closed-form optimal solution:
\begin{equation}\label{solution to d}
d_{k,i}^\star = \max\left(\tau_i,\min\left(\hat{d}_{k,i},\tau_{i+1}\right)\right),
\end{equation}
where
\begin{equation}
\begin{aligned}
&\hat{d}_{k,i}\\&=\frac{\sum\limits_{l\in\Gamma_i}(\tilde{a}_{k,l}+\chi_{k,l}^i)(\tilde{s}_{k,l}+1)+\sum\limits_{l\in\Omega_i}(\chi_{k,l}^i-\tilde{b}_{k,l})(\tilde{s}_{k,l}-1)}{\sum\limits_{l\in\Gamma_i}(\tilde{s}_{k,l}+1)^2+\sum\limits_{l\in\Omega_i}(\tilde{s}_{k,l}-1)^2}.
\end{aligned}
\end{equation}
Then, we have $d_k^\star=\mathop{\text{argmin}}\limits_{\left\{d_{k,i}^\star\right\}} \left\{p(d_{k,1}^\star),\dots,p(d_{k,\operatorname{card}(\mathcal{B})-1}^\star)\right\}$. This completes the proof.

\section{Proof of Theorem \ref{ET theorem I}}\label{appendix ET problem solution}
First note that the objective function in problem \eqref{ET subproblem I A} has one nonconvex term $-\operatorname{tr}\left(\mathbf{L}^H\mathbf{M}^{-1}\mathbf{L}\right)$ (the notation $(\mathbf{x})$ is dropped for simplicity of derivation), which can be upperbounded by the following first-order Taylor approximation \cite{ref303}:
\begin{equation}
	\begin{aligned}
		&\quad -\operatorname{tr}\left(\mathbf{L}^H\mathbf{M}^{-1}\mathbf{L}\right)\\&\le -2\Re\left\{\operatorname{tr}\left(\mathbf{L}_t^H\mathbf{M}_t^{-1}\mathbf{L}\right)\right\}+\operatorname{tr}\left(\mathbf{M}_t^{-1}\mathbf{L}_t\mathbf{L}_t^H\mathbf{M}_t^{-1}\mathbf{M}\right)+c_t,\label{extended Taylor lower bound}
	\end{aligned}
\end{equation}
where the subscript $(\cdot)_t$ indicates the value of the respective expression at the $t$-th iteration, and $c_t$ represents the constant term. Next, we recast the upper bound in \eqref{extended Taylor lower bound} by a more concise form. Specifically, by recalling that $\mathbf{L}= \tilde{\mathbf{X}}\mathbf{C}_{\mathbf{aa}}$ and $\tilde{\mathbf{X}}\triangleq\mathbf{X}^T\otimes \mathbf{I}_{N_r}$, we first express $\operatorname{tr}\left(\mathbf{L}_t^H\mathbf{M}_t^{-1}\mathbf{L}\right)$ as
\begin{equation}\label{extended lower bound term I}
	\begin{aligned}
		&\quad\operatorname{tr}\left(\mathbf{L}_t^H\mathbf{M}_t^{-1}\mathbf{L}\right)\\
		&=\operatorname{vec}\left(\left(\mathbf{C}_{\mathbf{aa}}\mathbf{L}_t^H\mathbf{M}_t^{-1}\right)^H\right)^H\operatorname{vec}\left(\mathbf{X}^T\otimes \mathbf{I}_{N_r}\right)\\
		&\overset{(a)}{=}\operatorname{vec}\left(\left(\mathbf{C}_{\mathbf{aa}}\mathbf{L}_t^H\mathbf{M}_t^{-1}\right)^H\right)^H
		\left(\mathbf{I}_{N_t}\otimes\mathbf{T}_{N_r,L}\otimes \mathbf{I}_{N_r}\right)\\&\quad \times
		\left(\left(\mathbf{T}_{N_t,L}\operatorname{vec}(\mathbf{X})\right)\otimes \operatorname{vec}\left(\mathbf{I}_{N_r}\right)\right)\\
		&\overset{(b)}{=}\operatorname{vec}\left(\left(\mathbf{C}_{\mathbf{aa}}\mathbf{L}_t^H\mathbf{M}_t^{-1}\right)^H\right)^H
		\left(\mathbf{I}_{N_t}\otimes\mathbf{T}_{N_r,L}\otimes \mathbf{I}_{N_r}\right)\\&\quad \times \left(\mathbf{T}_{N_t,L}\otimes\operatorname{vec}\left(\mathbf{I}_{N_r}\right)\right)\operatorname{vec}(\mathbf{X})\\&\triangleq \mathbf{l}_t^H\mathbf{x},
	\end{aligned}
\end{equation}
where we apply $\operatorname{vec}\left(\mathbf{A}\otimes \mathbf{B}\right)=\left(\mathbf{I}_n\otimes\mathbf{T}_{q,m}\otimes\mathbf{I}_{p}\right)\left(\operatorname{vec}(\mathbf{A})\otimes\operatorname{vec}(\mathbf{B})\right)$ for $\mathbf{A}\in\mathbb{C}^{m\times n}$ and $\mathbf{B}\in\mathbb{C}^{p\times q}$ in $(a)$ and $\left(\mathbf{A}\mathbf{B}\right)\otimes\left(\mathbf{CD}\right)=\left(\mathbf{A}\otimes\mathbf{C}\right)\left(\mathbf{B}\otimes\mathbf{D}\right)$ in $(b)$, and further define
$\mathbf{l}_t=\left(\mathbf{T}_{L,N_t}\otimes\operatorname{vec}\left(\mathbf{I}_{N_r}\right)\right)\times\left(\mathbf{I}_{N_t}\otimes\mathbf{T}_{L,N_r}\otimes \mathbf{I}_{N_r}\right)\operatorname{vec}\left(\left(\mathbf{C}_{\mathbf{aa}}\mathbf{L}_t^H\mathbf{M}_t^{-1}\right)^H\right)$.
Moreover, by substituting $\mathbf{M}=\tilde{\mathbf{X}}\mathbf{C}_{\mathbf{aa}}\tilde{\mathbf{X}}^H+\left(\frac{\pi}{2}-1\right)\operatorname{diag}\left(\tilde{\mathbf{X}}\mathbf{C}_{\mathbf{aa}}\tilde{\mathbf{X}}^H\right)+\frac{\pi\sigma_v^2}{2}\mathbf{I}_{N_rL}$ into $\operatorname{tr}\left(\mathbf{M}_t^{-1}\mathbf{L}_t\mathbf{L}_t^H\mathbf{M}_t^{-1}\mathbf{M}\right)$, we have
\begin{equation}\label{extended lower bound term II}
	\begin{aligned}
		&\quad\operatorname{tr}\left(\mathbf{M}_t^{-1}\mathbf{L}_t\mathbf{L}_t^H\mathbf{M}_t^{-1}\mathbf{M}\right)\\
		&=\operatorname{tr}\left(\mathbf{M}_t^{-1}\mathbf{L}_t\mathbf{L}_t^H\mathbf{M}_t^{-1}\tilde{\mathbf{X}}\mathbf{C}_{\mathbf{aa}}\tilde{\mathbf{X}}^H\right)\\
		&\quad+ \left(\frac{\pi}{2}-1\right)\operatorname{tr}\left(\mathbf{M}_t^{-1}\mathbf{L}_t\mathbf{L}_t^H\mathbf{M}_t^{-1}\operatorname{diag}\left(\tilde{\mathbf{X}}\mathbf{C}_{\mathbf{aa}}\tilde{\mathbf{X}}^H\right)\right)\\
		&\quad +\frac{\pi\sigma_v^2}{2}\operatorname{tr}\left(\mathbf{M}_t^{-1}\mathbf{L}_t\mathbf{L}_t^H\mathbf{M}_t^{-1}\right)\\
		&\overset{(a)}{=}\operatorname{tr}\left(\tilde{\mathbf{M}}_t\tilde{\mathbf{X}}\mathbf{C}_{\mathbf{aa}}\tilde{\mathbf{X}}^H\right)+\tilde{c}_t,\\
		&\overset{(b)}{=}\mathbf{x}^H\bar{\mathbf{M}}_t\mathbf{x} +\tilde{c}_t 
	\end{aligned}
\end{equation}
where
$\operatorname{tr}\left(\mathbf{A}\operatorname{diag}\left(\mathbf{B}\right)\right)=\operatorname{tr}\left(\operatorname{diag}\left(\mathbf{A}\right)\mathbf{B}\right)$ is applied in $(a)$ with $\tilde{\mathbf{M}}_t\triangleq \mathbf{M}_t^{-1}\mathbf{L}_t\mathbf{L}_t^H\mathbf{M}_t^{-1}+(\frac{\pi}{2}-1)\operatorname{diag}\left(\mathbf{M}_t^{-1}\mathbf{L}_t\mathbf{L}_t^H\mathbf{M}_t^{-1}\right)$ and $\tilde{c}_t$ being the constant term. Moreover, $(b)$ can be obtained in a similar way as \eqref{extended lower bound term I} with $\bar{\mathbf{M}}_t$ being
\begin{equation}\label{auxiliary variable Mt}
\begin{aligned}
\bar{\mathbf{M}}_t=\tilde{\mathbf{T}}^T
\left(\mathbf{C}_{\mathbf{aa}}^T\otimes \tilde{\mathbf{M}}_t\right)\tilde{\mathbf{T}}
\end{aligned}
\end{equation}
and $\tilde{\mathbf{T}}\triangleq \left(\mathbf{I}_{N_t}\otimes \mathbf{T}_{N_r,L}\otimes \mathbf{I}_{N_r}\right)\left(\mathbf{T}_{N_t,L}\otimes\operatorname{vec}\left(\mathbf{I}_{N_r}\right)\right)$. Note that \eqref{auxiliary variable Mt} can be evaluated efficiently by leveraging the commutativity of $\tilde{\mathbf{T}}$. 

Using \eqref{extended lower bound term I} and \eqref{extended lower bound term II}, the upper bound in \eqref{extended Taylor lower bound} can be equivalently rewritten as
\begin{equation}\label{extended Taylor lower bound transformation}
	\begin{aligned}
		-\operatorname{tr}\left(\mathbf{L}^H\mathbf{M}^{-1}\mathbf{L}\right)&\le - 2\Re\left\{\mathbf{l}_t^H\mathbf{x}\right\}+\mathbf{x}^H\bar{\mathbf{M}}_t\mathbf{x}+\bar{c}_t,
	\end{aligned}
\end{equation}
where $\bar{c}_t=c_t+\tilde{c}_t$. Therefore, a surrogate form of problem \eqref{ET subproblem I A} can be cast as
\begin{equation}\label{ET surrogate problem I}	
	\begin{aligned}
		\mathop{\text{minimize}}\limits_{\mathbf{x}}&\quad \mathbf{x}^H\bar{\mathbf{M}}_t\mathbf{x}- 2\Re\left\{\mathbf{l}_t^H\mathbf{x}\right\} +\rho\|\tilde{\mathbf{H}}\mathbf{x}-\mathbf{u}^i+\boldsymbol{\lambda}^i\|_2^2\\
		\text{subject to}&\quad \eqref{ISAC problem ET DAC III}.	
	\end{aligned}
\end{equation}
Furthermore, we apply \cite[Eq. (26)]{ref303} to upperbound the quadratic terms $\mathbf{x}^H\bar{\mathbf{M}}_t\mathbf{x}$ and $\|\tilde{\mathbf{H}}\mathbf{x}-\mathbf{u}^i+\boldsymbol{\lambda}^i\|_2^2$ in the objective of problem \eqref{ET surrogate problem I}, which yields
\begin{equation}\label{quadratic upper bound property}
\begin{aligned}	
&\quad \mathbf{x}^H\bar{\mathbf{M}}_t\mathbf{x}+\rho \|\tilde{\mathbf{H}}\mathbf{x}-\mathbf{u}^i+\boldsymbol{\lambda}^i\|_2^2\\
&\le \lambda_{\max}\left(\bar{\mathbf{M}}_t\right)\|\mathbf{x}\|_2^2+\rho \lambda_{\max}\left(\tilde{\mathbf{H}}^H\tilde{\mathbf{H}}\right)\|\mathbf{x}\|_2^2\\
&\quad +2\Re\left(\mathbf{x}_t^H\left(\bar{\mathbf{M}}_t-\lambda_{\max}\left(\bar{\mathbf{M}}_t\right)\mathbf{I}_{N_tL}\right)\mathbf{x}\right)\\&\quad-2\rho\Re\left(\left(\mathbf{u}^i-\boldsymbol{\lambda}^i\right)^H\tilde{\mathbf{H}}\mathbf{x}\right)\\
&\quad +2\rho\Re\left(\mathbf{x}_t^H\left(\tilde{\mathbf{H}}^H\tilde{\mathbf{H}}-\lambda_{\max}\left(\tilde{\mathbf{H}}^H\tilde{\mathbf{H}}\right)\mathbf{I}_{N_tL}\right)\mathbf{x}\right)+\hat{c}_t,
\end{aligned}
\end{equation}
where $\hat{c}_t$ denotes the irrelevant constant term. By using \eqref{quadratic upper bound property}, we thus obtain the surrogate of problem \eqref{ET surrogate problem I} as follows:
\begin{equation}\label{ET surrogate problem II}
\begin{aligned}
\mathop{\text{minimize}}\limits_{\mathbf{x}}&\quad \left(\lambda_{\max}\left(\bar{\mathbf{M}}_t\right)+\rho\lambda_{\max}\left(\tilde{\mathbf{H}}^H\tilde{\mathbf{H}}\right)\right)\|\mathbf{x}\|_2^2\\&\quad-2\Re\left(\mathbf{m}_t^H\mathbf{x}\right)\\
\text{subject to}&\quad \eqref{ISAC problem ET DAC III},	
\end{aligned}
\end{equation}
where
\begin{equation}\label{expression of mt}
\begin{aligned}	
\mathbf{m}_t&=\mathbf{l}_t+\left(\lambda_{\max}\left(\bar{\mathbf{M}}_t\right)\mathbf{I}_{N_tL}-\bar{\mathbf{M}}_t\right)\mathbf{x}_t +\rho\tilde{\mathbf{H}}^H\left(\mathbf{u}^i-\boldsymbol{\lambda}^i\right)\\&\quad+\rho\left(\lambda_{\max}\left(\tilde{\mathbf{H}}^H\tilde{\mathbf{H}}\right)\mathbf{I}_{N_tL}-\tilde{\mathbf{H}}^H\tilde{\mathbf{H}}\right)\mathbf{x}_t.
\end{aligned}
\end{equation}
Thus, the proof is completed.




\begin{thebibliography}{99}
\bibliographystyle{IEEEtran}

\bibitem{ref0}
Q.~Lin, H.~Shen, W.~Xu, and C.~Zhao, ``CRB oriented transmit waveform optimization for one-bit MIMO radar," in \emph{Proc. IEEE 101st Veh. Technol. Conf. (VTC2025-Spring)}, Oslo, Norway, Jun. 2025, pp.~1--6.

\bibitem{ref1}
F.~Liu \emph{et al.}, ``Joint radar and communication design: Applications, state-of-the-art, and the road ahead," \emph{IEEE Trans. Commun.}, vol.~68, no.~6, pp.~3834--3862, Jun.~2020.

\bibitem{ref2}
J.~A.~Zhang \emph{et al.}, ``An overview of signal processing techniques for joint communication and radar sensing," \emph{IEEE J. Sel. Topics Signal Process.}, vol.~15, no.~6, pp. 1295--1315, Nov.~2021.

\bibitem{ref3}
A.~Liu \emph{et al.}, ``A survey on fundamental limits of integrated sensing and communication," \emph{IEEE Commun. Surv. Tut.}, vol.~24, no.~2, pp. 994--1034, 2nd Quart., 2022.

\bibitem{ref4}
F.~Liu \emph{et al.}, ``Integrated sensing and communications: Toward dualfunctional wireless networks for 6G and beyond," \emph{IEEE J. Sel. Areas Commun.}, vol.~40, no.~6, pp. 1728--1767, Jun. 2022.

\bibitem{ref5}
\emph{Framework Overall Objectives Future Develop. IMT for 2030 Beyond}, ITU-R Standard M.2160-0, Nov. 2023.

\bibitem{ref21}
L.~Chen, F.~Liu, W.~Wang, and C.~Masouros, ``Joint radar-communication transmission: A generalized Pareto optimization framework," \emph{IEEE Trans. Signal Process.}, vol.~69, pp.~2752--2765, 2021.

\bibitem{ref22}
X.~Liu, T.~Huang, and Y.~Liu, ``Transmit design for joint MIMO radar and multiuser communications with transmit covariance constraint," \emph{IEEE J. Sel. Areas Commun.}, vol.~40, no.~ 6, pp. 1932--1950, Jun. 2022.

\bibitem{ref23}
R.~Liu, M.~Li, Q.~Liu, and A.~L.~Swindlehurst, ``Dual-functional radar-communication waveform design: A symbol-level precoding approach," \emph{IEEE J. Sel. Topics  Signal Process.}, vol.~15, no.~6, pp.~1316--1331, Nov. 2021.

\bibitem{ref24}
F.~Liu, Y.~-F.~Liu, A.~Li, C.~Masouros, and Y.~C.~Eldar, ``Cram\'er-Rao bound optimization for joint radar-communication beamforming," \emph{IEEE Trans. Signal Process.}, vol.~70, pp.~240--253, 2022.

\bibitem{ref25}
H.~Hua, T.~X.~Han, and J.~Xu, ``MIMO integrated sensing and communication: CRB-rate tradeoff," \emph{IEEE Trans. Wireless Commun.}, vol.~23, no.~4, pp. 2839--2854, Apr. 2024.

\bibitem{ref26}
C.~Xu and S.~Zhang, ``MIMO integrated sensing and communication exploiting prior information," \emph{IEEE J. Sel. Areas Commun.}, vol.~42, no.~9, pp. 2306--2321, Sep. 2024.

\bibitem{ref27}
C.~G.~Tsinos, A.~Arora, S.~Chatzinotas, and B.~Ottersten, ``Joint transmit waveform and receive filter design for dual-function radar-communication systems," \emph{IEEE J. Sel. Topics Signal Process.}, vol.~15, no.~6, pp.~1378--1392, Nov. 2021.

\bibitem{ref28}
L.~Chen, Z.~Wang, Y.~Du, Y.~Chen, and F.~R.~Yu, ``Generalized transceiver beamforming for DFRC with MIMO radar and MU-MIMO communication," \emph{IEEE J. Sel. Areas Commun.}, vol.~40, no.~6, pp.~1795--1808, Jun. 2022.

\bibitem{ref41}
S.~Jacobsson, G.~Durisi, M.~Coldrey, T.~Goldstein, and C.~ Studer, ``Quantized precoding for massive MU-MIMO," \emph{IEEE Trans. Commun.}, vol.~65, no.~11, pp.~4670--4684, Nov. 2017.

\bibitem{ref42}
M.~Shao, Q.~Li, W.~-K.~Ma, and A.~M.~-C.~So, ``A framework for one-bit and constant-envelope precoding over multiuser massive MISO channels," \emph{IEEE Trans. Signal Process.}, vol.~67, no.~20, pp.~5309--5324, Oct. 2019.

\bibitem{ref43}
S.~Cai, H.~Zhu, C.~Shen, and T.~-H.~Chang, ``Joint symbol level precoding and receive beamforming optimization for multiuser MIMO downlink," \emph{IEEE Trans. Signal Process.}, vol.~70, pp.~6185--6199, 2022.

\bibitem{ref44}
Y.~Li, C.~Tao, G.~Seco-Granados, A.~Mezghani, A.~L.~Swindlehurst, and L.~Liu, ``Channel estimation and performance analysis of one-bit massive MIMO systems," \emph{IEEE Trans. Signal Process.}, vol.~65, no.~15, pp. 4075--4089, Aug. 2017.

\bibitem{ref45}
M.~Shao and W.~-K.~Ma, ``Binary MIMO detection via homotopy optimization and its deep adaptation," \emph{IEEE Trans. Signal Process.}, vol.~69, pp.~781--796, 2021.

\bibitem{ref46}
Z.~Cheng, S.~Shi, Z.~He, and B.~Liao, ``Transmit sequence design for dual-function radar-communication system with one-bit DACs," \emph{IEEE Trans. Wireless Commun.}, vol.~20, no.~9, pp. 5846--5860, Sep. 2021.

\bibitem{ref47}
X.~Yu, Q.~Yang, Z.~Xiao, H.~Chen, V.~Havyarimana, and Z.~Han, ``A precoding approach for dual-functional radar-communication system with one-bit DACs," \emph{IEEE J. Sel. Areas Commun.}, vol.~40, no.~6, pp. 1965--1977, Jun. 2022.

\bibitem{ref48}
Q.~Lin, H.~Shen, Z.~Li, W.~Xu, C.~Zhao, and X. You, ``One-bit transceiver optimization for mmWave integrated sensing and communication systems," \emph{IEEE Trans. Commun.}, vol.~73, no.~2, pp. 800--816, Feb. 2025.

\bibitem{ref49}
K.~U.~Mazher, A.~Mezghani, and R.~W.~Heath, ``Improved CRB for millimeter-wave radar with 1-bit ADCs," \emph{IEEE Open J. Signal Process.}, vol.~2, pp.~318--335, May 2021.





\bibitem{ref102}
Y.~Liu, M.~Shao, W.~-K.~Ma, and Q.~Li, ``Symbol-level precoding through the lens of zero forcing and vector perturbation," \emph{IEEE Trans. Signal Process.}, vol.~70, pp. 1687--1703, 2022.

\bibitem{ref105}
Y.~Qin, A.~Li, Y.~Lyu, X.~Liao, and C.~Masouros, ``Symbol-Level precoding for PAPR reduction in multi-user MISO-OFDM systems," \emph{IEEE Trans. Wireless Commun.}, vol.~23, no.~9, pp.~12484--12498, Sep. 2024.

\bibitem{ref103} 
\"O.~T.~Demir and E.~Bj\"ornson, ``The Bussgang decomposition of nonlinear systems: Basic theory and MIMO extensions [lecture notes]," \emph{IEEE Signal Process. Mag.}, vol.~38, no.~1, pp.~131--136, Jan. 2021.

\bibitem{ref104}
A.~Mezghani and J.~A.~Nossek, ``Capacity lower bound of MIMO channels with output quantization and correlated noise," in \emph{Proc. IEEE Int. Symp. Inf. Theory (ISIT)}, Boston, MA, USA, Jul. 2012, pp. 1--5.


\bibitem{ref200}
M.~A.~Richards, \emph{Fundamentals of Radar Signal Processing}. New York, NY, USA: McGraw-Hill, 2014.

\bibitem{ref201}
S.~M.~Kay, \emph{Fundamentals of Statistical Signal Processing: Estimation Theory.} Englewood Cliffs, NJ, USA: Prentice-hall, 1993.

\bibitem{ref202}
P.~Stoica and P.~Babu, ``The Gaussian data assumption leads to the largest Cram\'er-Rao bound [Lecture Notes]," \emph{IEEE Signal Process. Mag.}, vol.~28, no.~3, pp. 132--133, May 2011.


\bibitem{ref204}
X.~Huang and B.~Liao, ``One-bit MUSIC," \emph{IEEE Signal Process. Lett.}, vol.~26, no.~7, pp. 961--965, Jul. 2019.

\bibitem{ref205}
M.~Ding, I.~Atzeni, A.~T\"olli, and A.~L.~Swindlehurst, ``On optimal MMSE channel estimation for one-bit quantized MIMO systems," \emph{IEEE Trans. Signal Process.}, vol.~73, pp.~617--632, 2025.

\bibitem{ref301}
S.~Boyd, N.~Parikh, E.~Chu, B~ Peleato, and J. Eckstein, ``Distributed optimization and statistical learning via the alternating direction method of multipliers," \emph{Found. Trends Mach. Learn.}, vol.~3, no.~1, pp. 1--122, 2011.

\bibitem{ref302}
L.~Li, X.~Wang, and G.~Wang, ``Alternating direction method of multipliers for separable convex optimization of real functions in complex variables," \emph{Math. Problems Eng.}, vol.~2015, Art.~no.~104531.

\bibitem{ref303}
Y.~Sun, P.~Babu, and D.~P.~Palomar, ``Majorization-minimization algorithms in signal processing, communications, and machine learning," \emph{IEEE Trans. Signal Process.}, vol.~65, no.~3, pp. 794--816, Feb. 2017.

\bibitem{ref304}
S.~Boyd and L.~Vandenberghe, \emph{Convex Optimization}. Cambridge, U.K.: Cambridge Univ. Press, 2004.

\bibitem{ref305}
Y.~-F.~Liu \emph{et al.}, ``A survey of recent advances in optimization methods for wireless communications," \emph{IEEE J. Sel. Areas Commun.}, vol.~42, no.~11, pp. 2992--3031, Nov. 2024.

\bibitem{ref306}
J. Nocedal and S. J. Wright, \emph{Numerical Optimization.} New York, NY, USA: Springer, 2006.


\bibitem{ref501}
S.~L.~Loyka, ``Channel capacity of MIMO architecture using the exponential correlation matrix," \emph{IEEE Commun. Lett.}, vol.~5, no.~9, pp. 369--371, Sep. 2001.

\bibitem{ref502}
J.~Li, L.~Xu, P.~Stoica, K.~W.~Forsythe, and D.~W.~Bliss, ``Range compression and waveform optimization for MIMO radar: A Cram\'er-Rao bound based study," \emph{IEEE Trans. Signal Process.}, vol.~56, no.~1, pp. 218--232, Jan. 2008.

\bibitem{ref503}
M. Grant and S. Boyd. (2014). \emph{CVX: MATLAB Software for Disciplined Convex Programming, Version 2.1.} [Online]. Available: http://cvxr.com/cvx/

\bibitem{ref504}
I.~Bekkerman and J.~Tabrikian, ``Target detection and localization using MIMO radars and sonars," \emph{IEEE Trans. Signal Process.}, vol.~54, no.~10, pp. 3873--3883, Oct. 2006.


\end{thebibliography}
\end{document}